\documentclass[lettersize,journal]{IEEEtran}
\usepackage{amsmath,amsfonts}
\usepackage{algorithmic}
\usepackage{algorithm}
\usepackage{array}
\usepackage{textcomp}
\usepackage{stfloats}
\usepackage{url}
\usepackage{verbatim}
\usepackage{graphicx}
\usepackage{cite}
\usepackage{amsthm,amsmath,amssymb}
\usepackage{mathrsfs}
\usepackage{enumitem}
\usepackage{subcaption} 
\usepackage{times}
\usepackage{svg}
\usepackage{multirow}
\usepackage{amsthm}            % load amsthm for theorem environments
\usepackage{longtable}
\usepackage{booktabs}
\usepackage{dsfont}

\usepackage{tabularx}
\usepackage{booktabs}

\usepackage{xcolor}
\usepackage[
  colorlinks=false,
  pdfborder={0 0 1},
  citebordercolor={0 1 0},
  linkbordercolor={1 0 0},
  urlbordercolor={0 1 0}
]{hyperref}

\theoremstyle{plain}
\newtheorem{theorem}{Theorem}

\begin{document}

\title{Coupling-Aware RHS Beamforming for Wideband Multi-User Sum Rate Maximization}

\author{Liangshun Wu, Wen Chen*
\thanks{Liangshun Wu  and  Wen Chen are with the Department of Electronic Engineering, Shanghai Jiao Tong University, Shanghai, China (e-mail: wuliangshun@sjtu.edu.cn; wenchen@sjtu.edu.cn).  }
}


\maketitle

\begin{abstract}
Wideband multi-user transmission assisted by reconfigurable holographic surfaces (RHSs) is fundamentally limited by mutual coupling effect among densely packed sub-wavelength radiation elements. This paper develops a coupling-aware wideband RHS model and an efficient joint beamforming framework to maximize the multi-user sum rate under practical feeder power and RHS excitation power constraints. We establish an electromagnetic equivalent model based on magnetic-dipole elements and a physically interpretable coupling decomposition into free space  near field coupling and guided surface wave coupling. For optimization, we employ a weighted minimum mean square error (WMMSE)-based block coordinate method with a closed-form digital precoder update and introduce a Jacobian-aided coupling consistent hologram update that preserves coupling sensitivity via a first-order surrogate while keeping the hologram subproblem convex and efficiently solvable by projected first-order methods. Meep experiments verify the correctness of the proposed coupling model, and the simulations for a 28~GHz, 1~GHz-bandwidth RHS downlink prove the effectiveness of Jacobian-aided WMMSE-based method.
\end{abstract}

\begin{IEEEkeywords}
RHS, Beamforming, Mutual Coupling, Sum Rate, WMMSE
\end{IEEEkeywords}

\section{Introduction}
\IEEEPARstart{R}{econfigurable}  holographic surfaces (RHSs) are a low cost alternative to large phased arrays for wideband high frequency links. With a few feeders generating a reference wave, densely packed sub wavelength meta elements modulate the field to form high gain beams using a large effective aperture, which is attractive at millimeter wave and terahertz bands. However, the large number of radiation elements makes the physical channel high dimensional, and pilot-based estimation becomes costly in wideband operation. This motivates beamforming oriented modeling \cite{10103817,9696209,11054266,11004435}, where an equivalent physics-based model maps feeder excitations and surface modulation to the multi-user channel across frequency. %,9696209,11054266,11004435

With element spacing as small as $1/4\lambda$ or less, dense packing leads to strong mutual coupling and the independent element assumption is no longer valid. Impedance-based coupling models widely used for phased arrays \cite{11274834}\cite{11274829}\cite{11123587} do not match the RHS excitation mechanism based on a feeder generated reference wave and surface modulation. Existing coupling-aware RHS models \cite{11175425} often rely on Green function terms without simple closed form expressions that directly capture the dependence on geometry, material parameters, and frequency, which is problematic in wideband systems where the coupling matrix must be recomputed for each subband. Large apertures cause near-/far-field coexistence (e.g., 30GHz, 256-element RHS has 80 m near field \cite{11195853}), which incurs the need for unified near-/far-field representation \cite{10819473,11195853,10901662,4232629,1632998,8798804}.

These electromagnetic and propagation effects are not only modeling issues but also fundamentally reshape the precoder and hologram optimization: the effective baseband channel depends on the RHS operator, which becomes frequency-dependent and strongly coupled across elements when mutual coupling is present. Consequently, the hologram profile and the digital precoders are intertwined through a nonlinear, wideband, and geometry or material-dependent mapping, making direct sum-rate maximization particularly challenging.

To handle multi-user interference, weighted minimum mean square error (WMMSE) provides an attractive route due to its tight  weighted sum rate (WSR) equivalence and efficient alternating updates. It was introduced for multi-user (MU) multiple-input multiple-output (MIMO) broadcast channel beamforming \cite{4712693} and later extended to multi-cell \cite{6399004,10778279} and interfering broadcast channels\cite{5756489}, yet not to hogographic MIMO in recent works\cite{11162320,11202841}. Moreover, with coupling-aware RHS modeling, the coupling matrix depends on holographic pattern implicitly through a coupling inverse, so the hologram update block is no longer a simple convex optimization problem under the standard freeze-the-operator approximation used in existing RHS optimizations \cite{11175425,11162320,11202841}. This approximation can be inaccurate under strong coupling and may lead to slow or unstable convergence, motivating a coupling-consistent update that explicitly accounts for the sensitivity of the coupled response to holographic pattern.

This paper addresses wideband RHS-assisted MU-MIMO transmission under mutual coupling and near-field propagation by developing a coupling-aware electromagnetic model, a unified channel representation, and a coupling-consistent joint beamforming algorithm. The main contributions are:
\begin{enumerate}[leftmargin=1.2em]
\item {Coupling-aware RHS equivalent model.}
We establish a mutual coupling aware electromagnetic equivalent model for RHSs based on an equivalent magnetic dipole description. 
The coupling matrix is explicitly characterized as a function of element geometry, medium parameters, and frequency, which supports subband-wise evaluation in wideband systems. 
To improve interpretability and calibration flexibility, we further decompose the coupling into free-space near-field coupling and surface-wave mediated coupling.

%\item {Unified near-/far-field channel representation via Weyl identity.}
%We develop a unified near- and far-field channel model using Weyl's identity, where a spherical wave is represented as an angular spectrum superposition of plane waves. 
%This yields a single formulation covering both regions and reduces basis mismatch for large apertures and short-range users.

\item {Constrained wideband sum-rate maximization with efficient WMMSE updates.}
We formulate a wideband multi-user sum-rate maximization over the holographic amplitude vector and the digital precoders across subbands under (i) a BS feeder sum-power constraint and (ii) an RHS excitation power constraint. 
We solve the resulting nonconvex coupled problem using a WMMSE-based block coordinate descent method. 
The digital precoder update admits a convex quadratic form with a Karush-Kuhn-Tucker (KKT) closed-form solution, where a single dual variable is efficiently found via bisection.

\item {Jacobian-aided coupling-consistent hologram update.}
To explicitly account for the nonlinear dependence of the coupled RHS response on the hologram, we propose a Jacobian-aided coupling-consistent hologram optimization step. 
By implicitly differentiating the coupled RHS operator,
we construct a first-order successive convex approximation (SCA) surrogate that retains the coupling feedback sensitivity term 
 rather than freezing the coupling inverse.
This leads to a convex quadratic hologram subproblem that can be solved efficiently by projected first-order methods, improving convergence stability and robustness in strong-coupling regimes.
\end{enumerate}

The paper is organized as follows. Section~II presents the wideband RHS model with mutual coupling and formulates the wideband multi-user rate maximization problem. Section~III derives the WMMSE solution. Section~IV provides simulations and results. Section~V concludes.

\section{ Modeling}
\subsection{RHS Beamforming in Wideband}
The total wideband bandwidth $B$ is divided into $U$ approximately flat subbands $\{\mathcal{W}_u\}_{u=1}^U$
with subband bandwidth $B_g \triangleq \frac{B}{U}$. Let $f_c$ denote the carrier frequency. The center frequency of the $u$-th subband is
$
f_u = f_c + \left(u-\frac{U+1}{2}\right) B_g,  u=1,\ldots,U .
$

Consider an RHS with $L$ feeders and $N$ radiation elements.
Let $\mathbf{r}_n$ denote the position vector of the $n$-th element, and $\mathbf{r}_l$ be that of the $l$-th feeder. Let
$
\mathbf{r}_n^{\,l} \triangleq \mathbf{r}_n-\mathbf{r}_l
$
be the distance vector from feeder $l$ to element $n$.

The object wave radiated toward user $k$ is modeled as
$
\Psi_{\mathrm{obj},k}(\mathbf{r}_n)
= \exp\!\left(-j\,\mathbf{k}_f(\theta_k,\varphi_k)\cdot \mathbf{r}_n\right),
$
where $\mathbf{k}_f(\theta_k,\varphi_k)$ is the free-space wavevector. 

The reference wave generated by feeder $l$ on the RHS surface is
$
\Psi_{\mathrm{ref},l}(\mathbf{r}_n^{\,l}; f_u)
=\exp\!\left(-j\,\mathbf{k}_{s,u}\cdot \mathbf{r}_n^{\,l}\right),
$
where $\mathbf{k}_{s,u}$ is the reference wave propagation
vector at subband center frequency $f_u$. 

Following the holographic interference principle, the interference term between the object wave (user $k$)
and the reference wave (feeder $l$) is
$
\Psi_{k,l}^{\mathrm{intf}}(\mathbf{r}_n)
=\Psi_{\mathrm{obj},k}(\mathbf{r}_n)\,\Psi_{\mathrm{ref},l}^{*}(\mathbf{r}_n^{\,l}; f_u).
$
where $\Psi_{\mathrm{ref},l}^{*}(\mathbf{r}_n^{\,l}; f_u)$ take the complex conjugate of $\Psi_{\mathrm{ref},l}(\mathbf{r}_n^{\,l}; f_u)$.
Then an amplitude-only hologram can be constructed as
$
m_{k,l}(\mathbf{r}_n)
=\frac{\Re\left\{\Psi_{k,l}^{\mathrm{intf}}(\mathbf{r}_n)\right\}+1}{2}\in[0,1].
$
where $\Re\!\left\{\Psi_{k,l}^{\mathrm{intf}}(\mathbf{r}_n)\right\}=\cos (\mathbf{k}_{s,u}\cdot \mathbf{r}_n^{\,l} - \mathbf{k}_f(\theta_k,\varphi_k)\cdot \mathbf{r}_n)$ denotes the real part of holographic pattern (see Fig. \ref{fig:rhs}).

For holographic pattern division multiple access (HDMA) \cite{deng2022hdma} superposition, the final normalized amplitude coefficient of element $n$ is
$m_n=\sum_{k=1}^{K}\sum_{l=1}^{L} a_{k,l}\,m_{k,l}(\mathbf{r}_n)$,
$a_{k,l}\ge 0$, $\sum_{k=1}^{K}\sum_{l=1}^{L} a_{k,l}=1$,
which guarantees $0\le m_n\le 1$.
Finally, define $D(\mathbf{m})=\mathrm{diag}(\mathbf{m})$ with $\mathbf{m}=[m_1,\ldots,m_N]^T$.

\begin{figure}[h]
\centering
\includegraphics[width=0.5\textwidth]{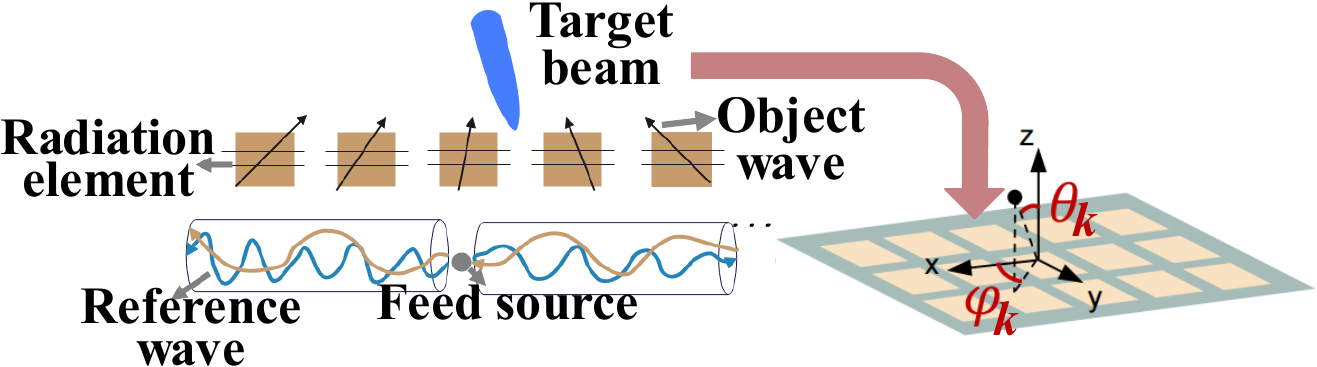}
\caption{Working principle of RHS.}
\label{fig:rhs}
\end{figure}

\subsection{Mutual Coupling }
\begin{figure*}[h]
\centering
\includegraphics[width=0.75\textwidth]{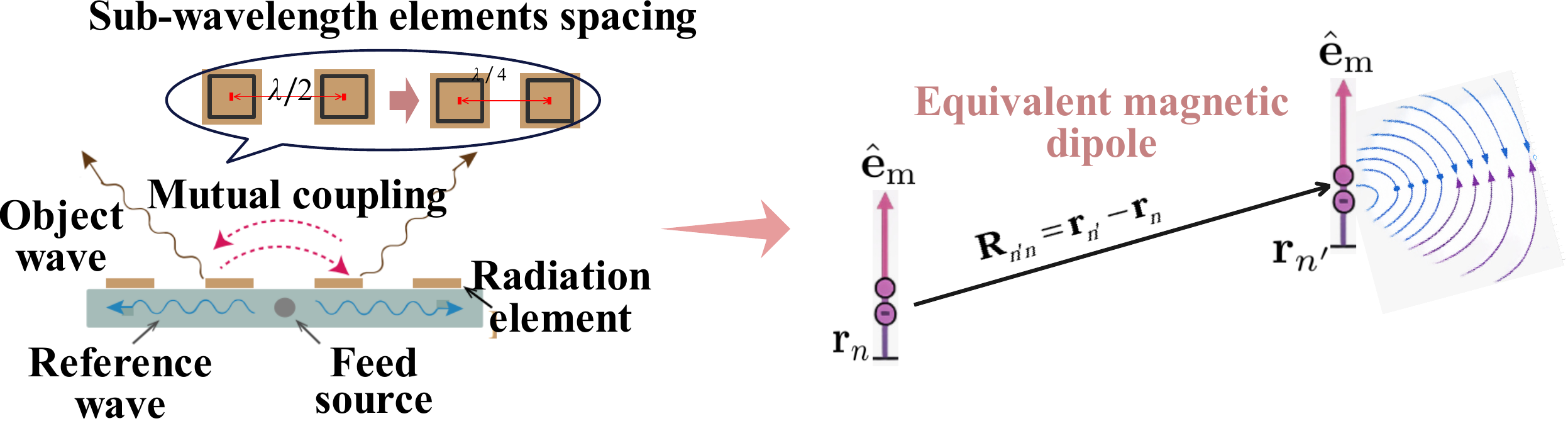}
\caption{Mutual coupling with equivalent dipole.}
\label{fig:coupling}
\end{figure*}
Each sub-wavelength RHS element is modeled as an equivalent magnetic dipole (see Fig. \ref{fig:coupling}) with a fixed orientation
$\hat{\mathbf{e}}_{\mathrm{m}}$. Let
$
\mathbf{p}_u \triangleq [p_{u,1},\ldots,p_{u,N}]^T \in \mathbb{C}^{N\times 1}
$
denote the complex dipole moment amplitudes at subband $u$ (scalar along $\hat{\mathbf{e}}_{\mathrm{m}}$).

Let $\mathbf{g}=[g_1,\ldots,g_N]^T$ be the element-wise complex polarizability, i.e., modulation coefficients, and
$
D(\mathbf{g})\triangleq \operatorname{diag}(\mathbf{g})$, $g_n = m_n e^{j\phi_n}.$
Amplitude-only holography corresponds to $\phi_n=0$, which means that it is approximately constant, hence $\mathbf{g}=\mathbf{m}$.

Let $\boldsymbol{\Xi}_u\in\mathbb{C}^{N\times N}$ be the mutual coupling matrix at frequency $f_u$, with $[\boldsymbol{\Xi}_u]_{n,n}=0$ absorbing self-terms into $\mathbf{g}$. Using an induced-dipole
superposition model,
$
\mathbf{p}_u = D(\mathbf{g})\left(\mathbf{f}_u^{\mathrm{ref}} + \boldsymbol{\Xi}_u \mathbf{p}_u\right),
$
Define the reference wave feeding matrix $\mathbf{F}_u\in\mathbb{C}^{N\times L}$ as
$
[\mathbf{F}_u]_{n,l} \triangleq \Psi_{\mathrm{ref},l}(\mathbf{r}_n^{\,l}; f_u),
$
the feeder-generated reference excitation is
$
\mathbf{f}_u^{\mathrm{ref}} = \mathbf{F}_u\mathbf{q}_u,
$
and $\mathbf{q}_u\in\mathbb{C}^{L\times 1}$ is the feeder input at subband $u$
(e.g., $\mathbf{q}_u = V_u \mathbf{s}_u$ with digital precoder $V_u$ and data symbols $\mathbf{s}_u$).
Thus,
$
\mathbf{p}_u
= \left(I_N - D(\mathbf{m})\boldsymbol{\Xi}_u\right)^{-1} D(\mathbf{m}) \mathbf{F}_u \mathbf{q}_u .
$

Therefore, the coupling-aware RHS equivalent matrix (mapping from feeder input to element dipole response)
is defined as
$
M_u(\mathbf{m})
\triangleq \left(I_N - D(\mathbf{m})\boldsymbol{\Xi}_u\right)^{-1} D(\mathbf{m}) \mathbf{F}_u
\in \mathbb{C}^{N\times L}.
$
Optionally define
$
C_u(\mathbf{m}) \triangleq \left(I_N - D(\mathbf{m})\boldsymbol{\Xi}_u\right)^{-1},
$
we have 
$
M_u(\mathbf{m}) = C_u(\mathbf{m})D(\mathbf{m})\mathbf{F}_u.
$

To reflect different coupling mechanisms, decompose
$
\boldsymbol{\Xi}_u = \boldsymbol{\Xi}_{\mathrm{fs},u} + \boldsymbol{\Xi}_{\mathrm{wg},u},
$
where $\boldsymbol{\Xi}_{\mathrm{fs},u}$ accounts for near-field coupling through free space (or a uniform medium),
and $\boldsymbol{\Xi}_{\mathrm{wg},u}$ accounts for guided or surface wave mediated coupling along the structure.

\subsubsection{Free-space near-field coupling $\boldsymbol{\Xi}_{\mathrm{fs},u}$}
Let $\mathbf{R}_{n'n}=\mathbf{r}_{n'}-\mathbf{r}_n$, $R_{n'n}=\|\mathbf{R}_{n'n}\|$, and
$\widehat{\mathbf{R}}_{n'n}=\mathbf{R}_{n'n}/R_{n'n}$. Using a homogeneous-medium approximation, the magnetic
field Green's function for a unit magnetic dipole oriented along $\hat{\mathbf{e}}_{\mathrm{m}}$ can be written as \eqref{eq:H},
\begin{figure*}
\begin{equation}
\begin{aligned}
\mathbf{H}(\mathbf{r}_{n'},\mathbf{r}_n; f_u)
= \frac{e^{-j k_u R_{n'n}}}{4\pi}
\Bigg[
k_u^{2}\Big(\widehat{\mathbf{R}}_{n'n}\times(\hat{\mathbf{e}}_{\mathrm{m}}\times\widehat{\mathbf{R}}_{n'n})\Big) 
 + \left(\frac{1}{R_{n'n}^2}-\frac{j k_u}{R_{n'n}}\right)
\left(3\widehat{\mathbf{R}}_{n'n}(\widehat{\mathbf{R}}_{n'n}\cdot \hat{\mathbf{e}}_{\mathrm{m}})-\hat{\mathbf{e}}_{\mathrm{m}}\right)
\Bigg]\frac{1}{R_{n'n}} .
\end{aligned}
\label{eq:H}
\end{equation}
\noindent\rule{\textwidth}{0.4pt}
\end{figure*}
where $k_u$ is the wavenumber of $f_u$:
$
k_u=2\pi f_u \sqrt{\mu\epsilon}
$ with medium parameters $(\mu,\epsilon)$.
In free space, it reduces to
$
k_u=\frac{2\pi f_u}{c_0}.
$

Eq. \eqref{eq:H} gives the homogeneous-medium Green's-function expression of the magnetic field
generated at $\mathbf{r}_{n'}$ by a {unit} magnetic dipole located at $\mathbf{r}_n$ and oriented along
$\hat{\mathbf{e}}_{\mathrm{m}}$. The prefactor $\exp(-j k_u R_{n'n})/(4\pi)$ represents the spherical-wave
propagation from the source element to the observation element, where $R_{n'n}$ is the separation distance
and $k_u$ is the (possibly medium-dependent) wavenumber at subband center frequency $f_u$. The remaining
vector term inside the brackets captures both the directional radiation pattern imposed by the dipole
orientation and the {distance-dependent} near-/intermediate-/far-field behaviors. In particular, the
first term $k_u^{2}\big(\widehat{\mathbf{R}}_{n'n}\times(\hat{\mathbf{e}}_{\mathrm{m}}\times\widehat{\mathbf{R}}_{n'n})\big)$
corresponds to the radiating  field component: using the vector identity
$\widehat{\mathbf{R}}\times(\hat{\mathbf{e}}_{\mathrm{m}}\times\widehat{\mathbf{R}})=
\hat{\mathbf{e}}_{\mathrm{m}}-(\hat{\mathbf{e}}_{\mathrm{m}}\!\cdot\!\widehat{\mathbf{R}})\widehat{\mathbf{R}}$,
it can be seen that this term extracts the component of the dipole orientation perpendicular to the
propagation direction $\widehat{\mathbf{R}}_{n'n}$, consistent with the fact that radiated fields are
transverse in the far field. With the outer factor $1/R_{n'n}$, this radiating contribution scales as
$\mathcal{O}(1/R_{n'n})$ and dominates at sufficiently large separations. The second term
$\big(\frac{1}{R_{n'n}^2}-\frac{j k_u}{R_{n'n}}\big)\big(3\widehat{\mathbf{R}}_{n'n}
(\widehat{\mathbf{R}}_{n'n}\!\cdot\!\hat{\mathbf{e}}_{\mathrm{m}})-\hat{\mathbf{e}}_{\mathrm{m}}\big)$
collects the induction and quasi-static components of the dipole field, whose magnitudes scale as
$\mathcal{O}(1/R_{n'n}^2)$ and $\mathcal{O}(1/R_{n'n}^3)$ after including the outer $1/R_{n'n}$ factor. These
terms are responsible for strong near-field interactions when elements are closely spaced (e.g.,
sub-wavelength spacing), which is precisely the regime where mutual coupling becomes non-negligible in RHS.
Moreover, the angular factor
$3\widehat{\mathbf{R}}_{n'n}(\widehat{\mathbf{R}}_{n'n}\!\cdot\!\hat{\mathbf{e}}_{\mathrm{m}})-\hat{\mathbf{e}}_{\mathrm{m}}$
(the ``3'' is explained in Appendix \ref{appA}) encodes the characteristic dipolar directional dependence.

A scalarized coupling coefficient (projected onto the dipole orientation) can be defined as
\begin{equation}
\left[\boldsymbol{\Xi}_{\mathrm{fs},u}\right]_{n',n}   = \hat{\mathbf{e}}_{\mathrm{m}}^{\top}\mathbf{H}(\mathbf{r}_{n'},\mathbf{r}_n; f_u),
 n'\neq n; 
\left[\boldsymbol{\Xi}_{\mathrm{fs},u}\right]_{n,n}=0.
\end{equation}
and define $  \xi_\mathrm{fs} \triangleq \sum{\left[\boldsymbol{\Xi}_{\mathrm{fs},u}\right]_{n',n}} / N$.

\subsubsection{Surface-wave mediated coupling $\boldsymbol{\Xi}_{\mathrm{wg},u}$}

A practical parameterized model is
\begin{equation}
\left[\boldsymbol{\Xi}_{\mathrm{wg},u}\right]_{n',n} =
\begin{cases}
\rho_u^{+}\exp\!\left(-(\alpha_{\mathrm{wg},u}+j\beta_{\mathrm{wg},u})\, s_{n'n}\right), & n'>n,\\
\rho_u^{-}\exp\!\left(-(\alpha_{\mathrm{wg},u}+j\beta_{\mathrm{wg},u})\, s_{n'n}\right), & n'<n,\\
0, & n'=n,
\end{cases}
\end{equation}
where $s_{n'n}$ is the equivalent guided propagation distance between elements $n$ and $n'$ along the feeding direction
(e.g., $s_{n'n}=|n'-n|d$ for a uniform linear array with spacing $d$). The coefficients $\rho_u^{+}$ and $\rho_u^{-}$ represent the coupling strength in the forward ($n'>n$) and reverse ($n'<n$) directions, respectively, which enables modeling possible asymmetry. $\exp(-\alpha_{\mathrm{wg},u}s_{n'n})$ captures the amplitude attenuation caused by guided-mode loss and leakage, while $\exp(-j\beta_{\mathrm{wg},u}s_{n'n})$ accounts for the phase accumulation of the guided wave, where $\alpha_{\mathrm{wg},u}$ and $\beta_{\mathrm{wg},u}$ are the effective attenuation and phase constants at subband center frequency $f_u$.
Define $  \xi_\mathrm{wg} \triangleq \sum{\left[\boldsymbol{\Xi}_{\mathrm{wg},u}\right]_{n',n}} / N$.

\subsection{Multi-User Sum Rate Maximization Problem}
\begin{figure}[h]
\centering
\includegraphics[width=0.5\textwidth]{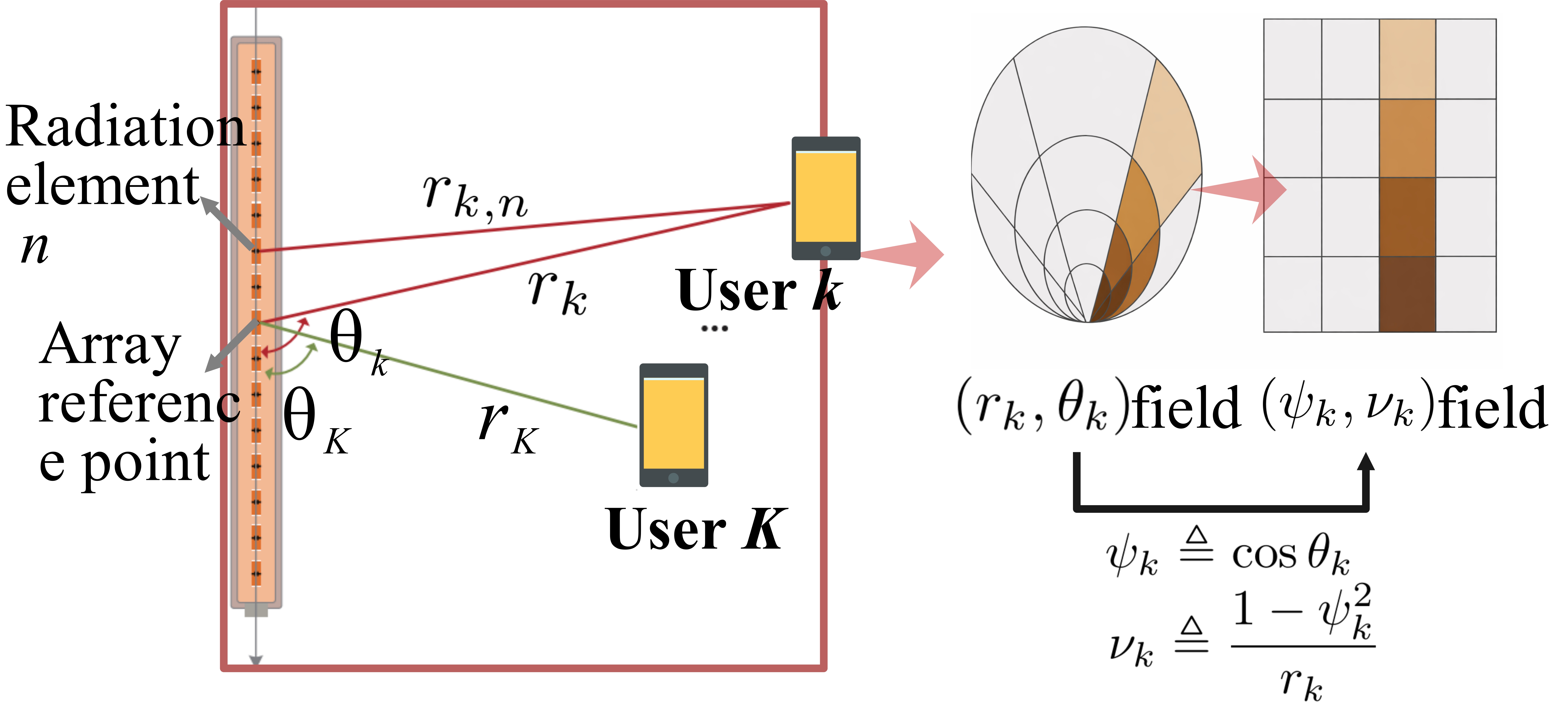}
\caption{Geometric relationship and $(r_k,\theta_k)$ to $(\psi_k,\nu_k)$ field conversion.}
\label{fig:geo}
\end{figure}

According to the geometric relationship (Fig.~\ref{fig:geo}), consider a uniform linear array (ULA) with inter-element spacing $d$.
Let the $n$-th element have relative position index $\delta_n$ (e.g., $\delta_n = n-\frac{N+1}{2}$),
so its coordinate along the array is $x_n=\delta_n d$.
Let user $k$ be located at distance $r_k$ from the array reference point, and define $\theta_k$ as
the angle such that the directional cosine w.r.t. the array axis is
$
\psi_k \triangleq \cos\theta_k.
$
Then the distance from user $k$ to element $n$ is 
$
r_{k,n} = \sqrt{r_k^2 + (\delta_n d)^2 - 2 r_k (\delta_n d)\psi_k } .
$
Using a second-order Taylor expansion for $|\delta_n d|\ll r_k$, we obtain
$
r_{k,n}
\approx r_k - \delta_n d \psi_k + \frac{\delta_n^2 d^2}{2 r_k}\left(1-\psi_k^2\right) 
= r_k - \delta_n d \psi_k + \frac{\delta_n^2 d^2}{2}\nu_k ,
$
where we define the curvature coefficient
$
\nu_k \triangleq \frac{1-\psi_k^2}{r_k}.
$
The far-field corresponds to $r_k\to\infty$, hence $\nu_k\to 0$ (see Fig. \ref{fig:near}).
\begin{figure}[h]
\centering
\includegraphics[width=0.28\textwidth]{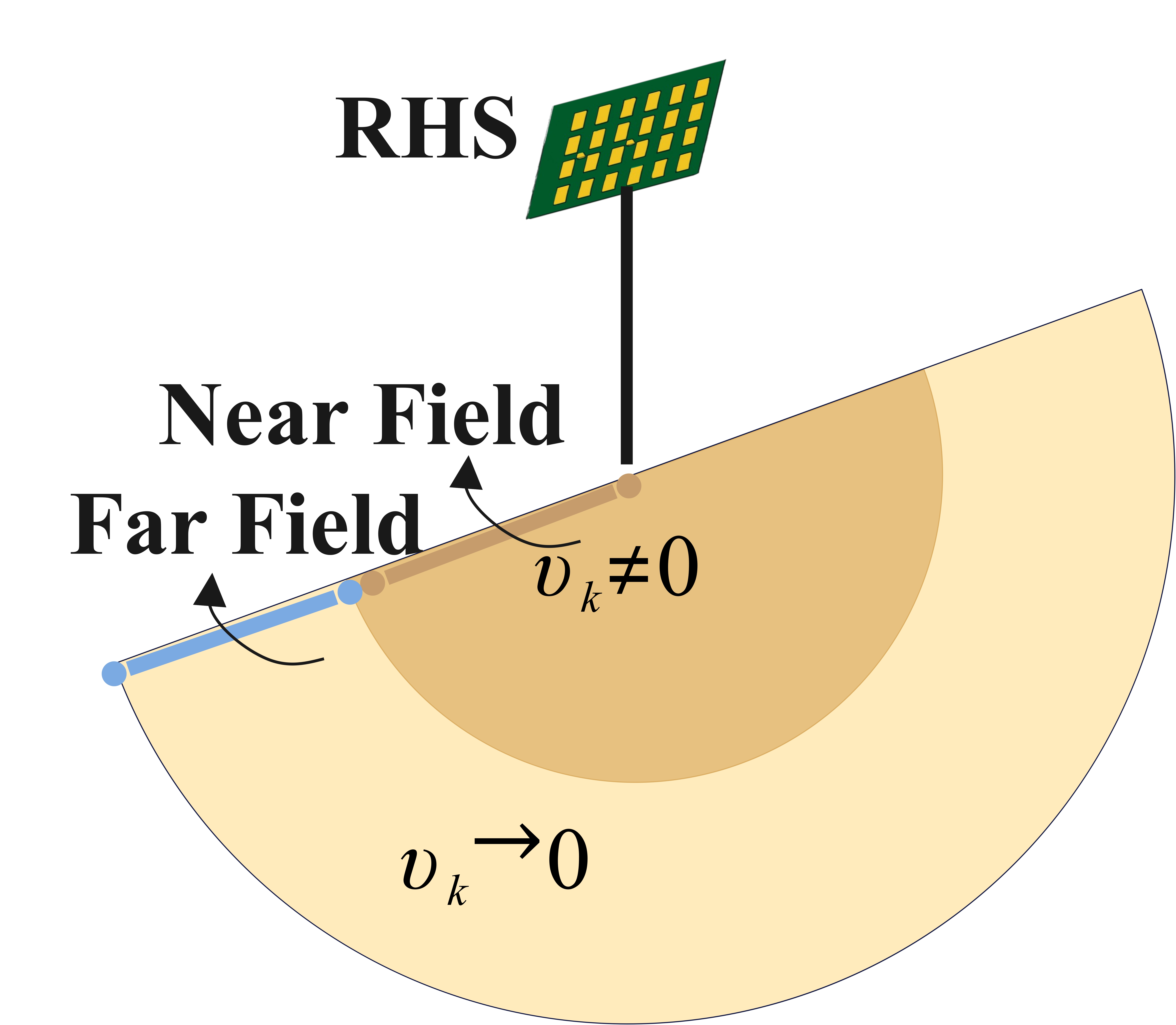}
\caption{Unified near- and far- filed representation with $(\psi_k,\nu_k)$: near-field, $\nu_k\neq 0$; far-field, $\nu_k\to 0$.}
\label{fig:near}
\end{figure}
\begin{figure*}[h]
\centering
\includegraphics[width=0.85\textwidth]{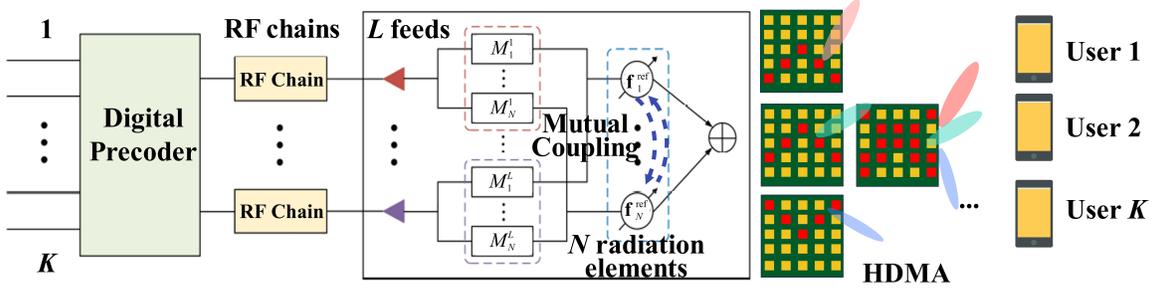}
\caption{RHS beamforming for multi users.}
\label{fig:rhs_beamforming}
\end{figure*}

Define the relative path length difference
$
\Delta r_{k,n} \triangleq r_{k,n}-r_k \approx - \delta_n d \psi_k + \frac{\delta_n^2 d^2}{2}\nu_k .
$
The row array response vector for the $u$-th subband is written as
$
\mathbf{b}_u(\psi_k,\nu_k)
= \frac{1}{\sqrt{N}}
\left[
e^{-j k_u \Delta r_{k,1}},\ldots,e^{-j k_u \Delta r_{k,N}}
\right]
\in\mathbb{C}^{1\times N}.
$
The LoS path gain is modeled as:
$
\beta_{k,u}
=
\frac{\lambda_u}{4\pi r_k}
\cdot
\exp\!\left(-\frac{\kappa_{\mathrm{abs}}(f_u) r_k}{2}\right)
\cdot
\exp\!\left(-j k_u r_k\right),
\lambda_u=\frac{c_0}{f_u},
$
where $\frac{\lambda_u}{4\pi r_k}$ denotes the spreading loss of amplitude and $\exp\left(-\frac{\kappa_{\mathrm{abs}}(f_u) r_k}{2}\right)$ denotes the molecular absorption loss of amplitude. $\kappa_{\mathrm{abs}}(f)$ is the molecular absorption coefficient of mmWave or Teraherz wave. (If one wants a more accurate near-field
amplitude variation, $r_k$ can be replaced by $r_{k,n}$ element-wise; here we use the common-amplitude approximation.)
The LoS channel row vector is then
$\mathbf{h}_{k,u}=\beta_{k,u}\sqrt{N}\,\mathbf{b}_u(\psi_k,\nu_k)\in\mathbb{C}^{1\times N}.$

The digital precoding matrix for the $u$-th subband is
$
V_u=\left[\mathbf{v}_{1,u},\ldots,\mathbf{v}_{K,u}\right]\in\mathbb{C}^{L\times K}.
$
Let the transmit symbol vector satisfy
$
\mathbb{E}\!\left[\mathbf{s}_u\mathbf{s}_u^H\right]=I_K$, $\mathbf{s}_u=[s_{1,u},\ldots,s_{K,u}]^T$.
Under the coupling-aware RHS equivalent model,
the element excitation (or equivalent dipole-response) vector is
$
\mathbf{x}_u = M_u(\mathbf{m})V_u\mathbf{s}_u \in\mathbb{C}^{N\times 1}.
$ The received signal for user $k$ on subband $u$ is written as 
$
y_{k,u}
= \mathbf{h}_{k,u}\mathbf{x}_u + n_{k,u} 
= \mathbf{h}_{k,u}M_u(\mathbf{m})\mathbf{v}_{k,u}s_{k,u}
+ \sum_{i\ne k}\mathbf{h}_{k,u}M_u(\mathbf{m})\mathbf{v}_{i,u}s_{i,u}
+ n_{k,u},
$
where $n_{k,u}$ follows $\mathcal{CN}(0,\sigma_{k,u}^2)$.
Define the effective baseband channel vector
$
\overline{\mathbf{h}}_{k,u}(\mathbf{m})
\triangleq
\mathbf{h}_{k,u}M_u(\mathbf{m})
\in\mathbb{C}^{1\times L}.
$
The SINR $\gamma_{k,u}$ is written as 
\begin{equation}
\begin{aligned}
\gamma_{k,u}(\mathbf{m},V_u)
&=
\frac{\left|\overline{\mathbf{h}}_{k,u}(\mathbf{m})\mathbf{v}_{k,u}\right|^2} 
{\sigma_{k,u}^2+\sum_{i\ne k}\left|\overline{\mathbf{h}}_{k,u}(\mathbf{m})\mathbf{v}_{i,u}\right|^2} \\
&= 
\frac{\left|\mathbf{h}_{k, u} M_u(\mathbf{m}) \mathbf{v}_{k, u}\right|^2}{\sigma_{k, u}^2+\sum_{i \neq k}\left|\mathbf{h}_{k, u} M_u(\mathbf{m}) \mathbf{v}_{i, u}\right|^2}\\
%&=
%\frac{\left|\mathbf{h}_{k, u} C_u(\mathbf{m})D(\mathbf{m})\mathbf{F}_u \mathbf{v}_{k, u}\right|^2}{\sigma_{k, u}^2+\sum_{i \neq k}\left|\mathbf{h}_{k, u} C_u(\mathbf{m})D(\mathbf{m})\mathbf{F}_u \mathbf{v}_{i, u}\right|^2}
\end{aligned}
\label{eq:sinr}
\end{equation}

The subband spectral efficiency is
$
\log_2\left(1+\gamma_{k,u}\right),
$
and the subband rate (in bit/s) is
$
R_{k,u}=B_g\log_2\left(1+\gamma_{k,u}\right).
$
The total sum rate is
$
R_{\mathrm{sum}}(\mathbf{m},\{V_u\})
=
\sum_{u=1}^U\sum_{k=1}^K R_{k,u}.
$

Under the mutual coupling-aware RHS equivalent model, the joint optimization to maximize the total sum rate
over the holographic amplitude $\mathbf{m}$ and subband precoders $\{V_u\}$ is
\begin{gather}
\max_{\mathbf{m},\{V_u\}}
\quad
\sum_{u=1}^U\sum_{k=1}^K B_g\log_2\left(1+\gamma_{k,u}(\mathbf{m},V_u)\right)
\\
\text{s.t.}\quad
\sum_{u=1}^U \mathrm{Tr}\!\left(V_uV_u^H\right)\le P_{\mathrm{BS}},
\\
0\le m_n\le 1,\ \forall n,
\\
\eta\sum_{u=1}^U \mathrm{Tr}\!\left(M_u(\mathbf{m})V_uV_u^H M_u(\mathbf{m})^H\right)\le P_{\mathrm{RHS}},
\end{gather}
where $P_{\mathrm{BS}}$ is the total feeder  power budget across all subbands,
$P_{\mathrm{RHS}}$ is the RHS element excitation power budget, and $\eta\in(0,1]$
is an efficiency factor.

\section{Solution}
Since $\gamma_{k,u}$ is a quadratic-over-quadratic function and $\mathbf{m}$ is multiplicatively coupled with ${V_u}$ via $M_u(\mathbf{m})$, we adopt a weighted minimum mean square error (WMMSE)-based block-coordinate strategy that alternately updates ${V_u}$ and $\mathbf{m}$, where each subproblem becomes a convex quadratic programming(QP) or quadratically constrained quadratic programming (QCQP).

\subsection{WMMSE Equivalence}
For each user--subband pair $(k,u)$, introduce a scalar linear equalizer $g_{k,u}\in\mathbb{C}$ and estimate
the transmitted symbol as
$
\hat{s}_{k,u}=g_{k,u}^*\,y_{k,u}\in\mathbb{C}^{1\times L}.
$

Denote the mean squared error (MSE) as  $e_{k,u} =  \mathbb{E}\left[\left|\hat{{s}}_{k, u}-s_{k, u}\right|^2\right]$. Introduce a weight $w_{k,u}>0$ and define the weighted MSE (WMSE) cost
$
\xi_{k,u}\!\left(w_{k,u},g_{k,u},\mathbf{m},V_u\right)
\triangleq
w_{k,u}e_{k,u}-\ln w_{k,u},
$
where $\ln(\cdot)$ is the natural logarithm. Note that $\log_2(1+\gamma_{k,u})=\ln(1+\gamma_{k,u})/\ln 2$; thus using
$\ln(\cdot)$ only introduces a constant scaling factor with respect to base-2 rates.

\begin{theorem}[WSR--WMMSE equivalence]
For any fixed $(\mathbf{m},\{V_u\})$, consider the MSE $e_{k,u}$ associated with the linear equalizer $g_{k,u}\in\mathbb{C}$ and the WMSE cost
$\xi_{k,u}$.
Then the minimizers of $\xi_{k,u}$ are given by
$
g_{k,u}^{\star}
=
\frac{\overline{\mathbf{h}}_{k,u}(\mathbf{m})\,\mathbf{v}_{k,u}}
{\sum_{i=1}^{K}\left|\overline{\mathbf{h}}_{k,u}(\mathbf{m})\,\mathbf{v}_{i,u}\right|^2+\sigma_{k,u}^2},
$
$
e_{k,u}^{\star}=\frac{1}{1+\gamma_{k,u}},
$
$
w_{k,u}^{\star}= \frac{1}{e_{k,u}^{\star}}=1+\gamma_{k,u}.
$
Moreover,
$
\min_{g_{k,u},\,w_{k,u}>0}\xi_{k,u}
=
1-\ln\!\left(1+\gamma_{k,u}\right),
$
and thus maximizing weighted sum rate (WSR) $\sum_{k,u}\log_2(1+\gamma_{k,u})$ is equivalent to minimizing $\sum_{k,u}\xi_{k,u}$ with respect to
$(\mathbf{m},\{V_u\})$.
\end{theorem}

\begin{proof}
Recall that the received signal is
$
y_{k,u}=\overline{\mathbf{h}}_{k,u}(\mathbf{m})V_u\mathbf{s}_u+n_{k,u},
$
with $\mathbb{E}[\mathbf{s}_u\mathbf{s}_u^H]=\mathbf{I}_K$ and $n_{k,u}\sim\mathcal{CN}(0,\sigma_{k,u}^2)$.
Using the linear estimate $\hat{s}_{k,u}=g_{k,u}^*y_{k,u}$, the MSE is
$e_{k,u}=\mathbb{E}\!\left[\left|g_{k,u}^*y_{k,u}-s_{k,u}\right|^2\right]$.
Expanding it and using the mutual independence of $\{s_{i,u}\}$ and $n_{k,u}$ with
$\mathbb{E}[|s_{i,u}|^2]=1$ and $\mathbb{E}[|n_{k,u}|^2]=\sigma_{k,u}^2$, $\mathbb{E}\left[s_{i, u} s_{j, u}^*\right]=0, i \neq j$ and $\mathbb{E}\left[s_{i, u} n_{k, u}^*\right]=0$, hence all cross terms vanish, yielding eq. \eqref{eq:mse}.
\begin{figure*}
\begin{equation}
\begin{aligned}
e_{k, u} & =\mathbb{E}\left[\left|g_{k,u}^* y_{k, u}-s_{k, u}\right|^2\right] \\
& =\mathbb{E}\left[\left|g_{k,u}^*\left(\overline{\mathbf{h}}_{k, u}(\mathbf{m}) V_u \mathbf{s}_u+n_{k, u}\right)-s_{k, u}\right|^2\right] \\
& =\mathbb{E}\left[\left|g_{k,u}^* \overline{\mathbf{h}}_{k, u}(\mathbf{m}) \mathbf{v}_{k, u} s_{k, u}+\sum_{i \neq k} g_{k,u}^* \overline{\mathbf{h}}_{k, u}(\mathbf{m}) \mathbf{v}_{i, u} s_{i, u}+g_{k,u}^* n_{k, u}-s_{k, u}\right|^2\right] \\
&=\mathbb{E}\Big[|g_{k,u}^* \overline{\mathbf{h}}_{k,u}(\mathbf{m}) \mathbf{v}_{k,u}s_{k,u}|^2\Big]
+\mathbb{E}\Big[\Big|\sum_{i\neq k} g_{k,u}^*\overline{\mathbf{h}}_{k,u}(\mathbf{m})\mathbf{v}_{i,u}s_{i,u}\Big|^2\Big]
+\mathbb{E}\big[|g_{k,u}^* n_{k,u}|^2\big] -2\Re\Big\{\mathbb{E}\big[g_{k,u}^* \overline{\mathbf{h}}_{k,u}(\mathbf{m})\mathbf{v}_{k,u}s_{k,u}s_{k,u}^*\big]\Big\}\\
&\quad 
+ \mathbb{E}[|s_{k,u}|^2]\\
& =\left|g_{k,u}\right|^2 \sum_{i=1}^K\left|\overline{\mathbf{h}}_{k, u}(\mathbf{m}) \mathbf{v}_{i, u}\right|^2+\left|g_{k,u}\right|^2 \sigma_{k, u}^2-2 \Re\left\{g_{k,u} \overline{\mathbf{h}}_{k, u}(\mathbf{m}) \mathbf{v}_{k, u}\right\}+1\\
&=\left(\sum_{i=1}^{K}\left|\overline{\mathbf{h}}_{k,u}(\mathbf{m})\,\mathbf{v}_{i,u}\right|^2+\sigma_{k,u}^2\right)
\Big(|g_{k,u}|^2-2\Re\left\{g_{k,u}\tfrac{\overline{\mathbf{h}}_{k,u}(\mathbf{m})\,\mathbf{v}_{k,u}}
{\sum_{i=1}^{K}\left|\overline{\mathbf{h}}_{k,u}(\mathbf{m})\,\mathbf{v}_{i,u}\right|^2+\sigma_{k,u}^2}\right\}\Big)
+1\\
&=\left(\sum_{i=1}^{K}\left|\overline{\mathbf{h}}_{k,u}(\mathbf{m})\,\mathbf{v}_{i,u}\right|^2+\sigma_{k,u}^2\right)
\left|\,g_{k,u}-\frac{\overline{\mathbf{h}}_{k,u}(\mathbf{m})\,\mathbf{v}_{k,u}}
{\sum_{i=1}^{K}\left|\overline{\mathbf{h}}_{k,u}(\mathbf{m})\,\mathbf{v}_{i,u}\right|^2+\sigma_{k,u}^2}\right|^2 
+\,1-\frac{\left|\overline{\mathbf{h}}_{k,u}(\mathbf{m})\,\mathbf{v}_{k,u}\right|^2}
{\sum_{i=1}^{K}\left|\overline{\mathbf{h}}_{k,u}(\mathbf{m})\,\mathbf{v}_{i,u}\right|^2+\sigma_{k,u}^2}.
\end{aligned}
\label{eq:mse}
\end{equation}
\noindent\rule{\textwidth}{0.4pt}
\end{figure*}

From eq. \eqref{eq:mse}, the first term is a nonnegative squared term, hence the minimum is achieved when this squared term equals zero, which directly yields the closed-form solution:
\begin{equation}
g_{k,u}^{\star}
=
\frac{\overline{\mathbf{h}}_{k,u}(\mathbf{m})\,\mathbf{v}_{k,u}}
{\sum_{i=1}^{K}\left|\overline{\mathbf{h}}_{k,u}(\mathbf{m})\,\mathbf{v}_{i,u}\right|^2+\sigma_{k,u}^2},
\end{equation}

For fixed $(\mathbf{m},V_u)$, the function in \eqref{eq:mse} is strictly convex in $g_{k,u}$.
Taking the Wirtinger derivative and setting $\partial e_{k,u}/\partial g_{k,u}^*=0$ gives 
\begin{equation}
\begin{aligned}
e_{k,u}^{\star}
&=
1-
\frac{\left|\overline{\mathbf{h}}_{k,u}(\mathbf{m})\,\mathbf{v}_{k,u}\right|^2}
{\sum_{i=1}^{K}\left|\overline{\mathbf{h}}_{k,u}(\mathbf{m})\,\mathbf{v}_{i,u}\right|^2+\sigma_{k,u}^2} \\
&=
\frac{
\sum_{i=1}^{K}\left|\overline{\mathbf{h}}_{k,u}(\mathbf{m})\,\mathbf{v}_{i,u}\right|^2+\sigma_{k,u}^2
-
\left|\overline{\mathbf{h}}_{k,u}(\mathbf{m})\,\mathbf{v}_{k,u}\right|^2
}
{\sum_{i=1}^{K}\left|\overline{\mathbf{h}}_{k,u}(\mathbf{m})\,\mathbf{v}_{i,u}\right|^2+\sigma_{k,u}^2} \\
&=
\frac{
\sigma_{k,u}^2+\sum_{i\neq k}\left|\overline{\mathbf{h}}_{k,u}(\mathbf{m})\,\mathbf{v}_{i,u}\right|^2
}
{\sigma_{k,u}^2+\sum_{i=1}^{K}\left|\overline{\mathbf{h}}_{k,u}(\mathbf{m})\,\mathbf{v}_{i,u}\right|^2}
=
\frac{1}{1+\gamma_{k,u}}.
\end{aligned}
\label{eq:emmse_proof}
\end{equation}

Fix $e>0$ and consider $\phi(w)=we-\ln w$, $w>0$. This function is strictly convex in $w$. Setting $\phi'(w)=e-\frac{1}{w}=0$ yields
$
w^\star=\frac{1}{e}.
$
Therefore, with $e=e_{k,u}^{\star}$,
$
w_{k,u}^\star=\frac{1}{e_{k,u}^{\star}}=1+\gamma_{k,u}.
$

Therefore, we have eq. \eqref{eq:wmse_min_long},
\begin{figure*}
\begin{equation}
\begin{aligned}
\min_{g_{k,u},\,w_{k,u}>0}\ \xi_{k,u}
&=\min_{g_{k,u},\,w_{k,u}>0}\Big(w_{k,u}e_{k,u}-\ln w_{k,u}\Big)
=\min_{g_{k,u}}\ \min_{w_{k,u}>0}\Big(w_{k,u}e_{k,u}-\ln w_{k,u}\Big) \\
&=\min_{g_{k,u}}\Bigg(\underbrace{\min_{w_{k,u}>0}\Big(w_{k,u}e_{k,u}-\ln w_{k,u}\Big)}_{\ \frac{\partial}{\partial w_{k,u}}:\ e_{k,u}-\frac{1}{w_{k,u}}=0\ \Rightarrow\ w_{k,u}^\star=\frac{1}{e_{k,u}}\ }\Bigg)
=\min_{g_{k,u}}\Big( w_{k,u}^\star e_{k,u}-\ln w_{k,u}^\star \Big) \\
&=\min_{g_{k,u}}\Big( 1+\ln e_{k,u}\Big)
=1+\ln\!\Big(\min_{g_{k,u}} e_{k,u}\Big)
=1+\ln\!\left(e_{k,u}^{\star}\right)
=1+\ln\!\left(\frac{1}{1+\gamma_{k,u}}\right)
=1-\ln\!\left(1+\gamma_{k,u}\right).
\end{aligned}
\label{eq:wmse_min_long}
\end{equation}
\noindent\rule{\textwidth}{0.4pt}
\end{figure*}
which completes the proof.
\end{proof}

We solve the WMMSE reformulation
\begin{equation}
\min_{\mathbf{m},\{V_u\},\{g_{k,u},w_{k,u}\}}
\ \sum_{u=1}^{U}\sum_{k=1}^{K}\left(w_{k,u}e_{k,u}-\ln w_{k,u}\right)
\label{eq:wmmse_problem}
\end{equation}
For fixed $(\mathbf{m},\{V_u\})$, the MMSE equalizer and weight admit closed-form updates
$
g_{k,u}\leftarrow
\frac{\overline{\mathbf{h}}_{k,u}(\mathbf{m})\,\mathbf{v}_{k,u}}
{\sum_{i=1}^{K}\left|\overline{\mathbf{h}}_{k,u}(\mathbf{m})\,\mathbf{v}_{i,u}\right|^2+\sigma_{k,u}^2}$,
$w_{k,u}\leftarrow \frac{1}{e_{k,u}}$.
Then we apply block-coordinate descent: updating $\{V_u\}$ with $\mathbf{m}$ fixed yields a convex QP/QCQP,
while updating $\mathbf{m}$ becomes a convex QP/QCQP after freezing the term   $M_u(\mathbf{m})$.

\subsection{Digital Precoding Subproblem for $\{\mathbf{V}_u\}$}
\label{subsec:V_update}

We update the digital precoders $\{\mathbf{V}_u\}$ while fixing the holographic amplitude $\mathbf{m}$ and the WMMSE auxiliary
variables $\{g_{k,u},w_{k,u}\}$. In this case, the effective channels
$\overline{\mathbf{h}}_{k,u}(\mathbf{m})=\mathbf{h}_{k,u}\mathbf{M}_u(\mathbf{m})$ are constants, and the $\{\mathbf{V}_u\}$-dependent
part of the WMMSE objective becomes a convex quadratic function.

\begin{theorem}[Convex QP and KKT closed-form update for $\{\mathbf{V}_u\}$]
\label{thm:V_update_short}
Fix $\mathbf{m}$ and $\{g_{k,u},w_{k,u}\}$. Define the effective channel column vector
$\tilde{\mathbf{h}}_{k,u}(\mathbf{m})\triangleq \overline{\mathbf{h}}_{k,u}(\mathbf{m})^{H}\in\mathbb{C}^{L\times 1}$, and for each subband $u$,
$\mathbf{A}_u(\mathbf{m}) \triangleq \sum_{k=1}^{K} w_{k,u}|g_{k,u}|^2\,\tilde{\mathbf{h}}_{k,u}(\mathbf{m})\tilde{\mathbf{h}}_{k,u}(\mathbf{m})^{H}\succeq 0$, $
\mathbf{B}_u(\mathbf{m}) \triangleq \big[w_{1,u}g_{1,u}^{*}\tilde{\mathbf{h}}_{1,u}(\mathbf{m}),\ldots,w_{K,u}g_{K,u}^{*}\tilde{\mathbf{h}}_{K,u}(\mathbf{m})\big]$.
Then the digital precoding subproblem is a convex QP with a single sum-power constraint, and its KKT point satisfies
\begin{equation}
\mathbf{V}_u(\lambda)=\big(\mathbf{A}_u(\mathbf{m})+\lambda\mathbf{I}_L\big)^{-1}\mathbf{B}_u(\mathbf{m}),\quad u=1,\ldots,U,
\end{equation}
where $\lambda\ge 0$ is chosen such that $\sum_{u=1}^U\mathrm{Tr}\big(\mathbf{V}_u(\lambda)\mathbf{V}_u(\lambda)^H\big)\le P_{\mathrm{BS}}$ with complementary slackness.
Moreover, $\sum_{u=1}^U\mathrm{Tr}\big(\mathbf{V}_u(\lambda)\mathbf{V}_u(\lambda)^H\big)$ is monotonically decreasing in $\lambda$, hence $\lambda$ can be found by bisection.
\end{theorem}

\begin{proof}
Since $\mathbf{m}$ is fixed, $\overline{\mathbf{h}}_{k,u}(\mathbf{m})$ is fixed. Starting from the MSE expression under a given equalizer $g_{k,u}$, and plugging 
\begin{equation}
\begin{aligned}
e_{k,u}
=&\mathbb{E}\!\left[\left|g_{k,u}^{*}y_{k,u}-s_{k,u}\right|^2\right]\\
=&\left|g_{k,u}\right|^2\!\left(\sum_{i=1}^{K}\left|\overline{\mathbf{h}}_{k,u}(\mathbf{m})\,\mathbf{v}_{i,u}\right|^2+\sigma_{k,u}^2\right) \\
&-2\Re\!\left\{g_{k,u}\overline{\mathbf{h}}_{k,u}(\mathbf{m})\,\mathbf{v}_{k,u}\right\}+1,
\end{aligned}
\label{eq:Vproof_mse}
\end{equation}
into the WMMSE sum objective $\sum_{u,k}\big(w_{k,u}e_{k,u}-\ln w_{k,u}\big)$, we drop constants independent of $\{\mathbf{V}_u\}$.
Using $\tilde{\mathbf{h}}_{k,u}(\mathbf{m})=\overline{\mathbf{h}}_{k,u}(\mathbf{m})^{H}$ and collecting quadratic terms yields the standard trace form:
\begin{figure*}[t]
\begin{equation}
\begin{aligned}
\min_{\{\mathbf{V}_u\}}
\quad &
\sum_{u=1}^{U}\sum_{k=1}^{K} w_{k,u}|g_{k,u}|^2
\sum_{i=1}^{K}\big|\overline{\mathbf{h}}_{k,u}(\mathbf{m})\mathbf{v}_{i,u}\big|^2
-2\sum_{u=1}^{U}\sum_{k=1}^{K} w_{k,u}\Re\!\left\{g_{k,u}\overline{\mathbf{h}}_{k,u}(\mathbf{m})\mathbf{v}_{k,u}\right\} \\
\text{s.t.}\quad &
\sum_{u=1}^{U}\mathrm{Tr}\!\left(\mathbf{V}_u\mathbf{V}_u^{H}\right)\le P_{\mathrm{BS}}\\
\Longleftrightarrow\quad &
\min_{\{\mathbf{V}_u\}}
\ \sum_{u=1}^{U}\Big(
\mathrm{Tr}\!\left(\mathbf{V}_u^{H}\mathbf{A}_u(\mathbf{m})\mathbf{V}_u\right)
-2\Re\!\left\{\mathrm{Tr}\!\left(\mathbf{V}_u^{H}\mathbf{B}_u(\mathbf{m})\right)\right\}
\Big)
\quad \text{s.t.}\quad
\sum_{u=1}^{U}\mathrm{Tr}\!\left(\mathbf{V}_u\mathbf{V}_u^{H}\right)\le P_{\mathrm{BS}},
\end{aligned}
\label{eq:Vproof_qp}
\end{equation}
\noindent\rule{\textwidth}{0.4pt}
\end{figure*}
where $\mathbf{A}_u(\mathbf{m})\succeq 0$ is a nonnegative weighted sum of rank-one PSD matrices, hence the objective is convex quadratic and the feasible set is convex.
Introduce the Lagrange multiplier $\lambda\ge 0$ for the BS sum-power constraint. The Lagrangian is
\begin{figure*}[t]
\begin{equation}
\begin{aligned}
\mathcal{L}(\{\mathbf{V}_u\},\lambda)
&=
\sum_{u=1}^{U}\Big(
\mathrm{Tr}\!\left(\mathbf{V}_u^{H}\mathbf{A}_u\mathbf{V}_u\right)
-2\Re\!\left\{\mathrm{Tr}\!\left(\mathbf{V}_u^{H}\mathbf{B}_u\right)\right\}
\Big)
+\lambda\!\left(\sum_{u=1}^{U}\mathrm{Tr}\!\left(\mathbf{V}_u\mathbf{V}_u^{H}\right)-P_{\mathrm{BS}}\right),\\
\frac{\partial \mathcal{L}}{\partial \mathbf{V}_u^{*}}=\mathbf{0}
&\Longleftrightarrow
\left(\mathbf{A}_u+\lambda\mathbf{I}_L\right)\mathbf{V}_u=\mathbf{B}_u
\Longleftrightarrow
\mathbf{V}_u(\lambda)=\left(\mathbf{A}_u+\lambda\mathbf{I}_L\right)^{-1}\mathbf{B}_u,\quad \forall u,\\
&\text{with}\quad
\lambda\ge 0,\quad
\sum_{u=1}^{U}\mathrm{Tr}\!\left(\mathbf{V}_u(\lambda)\mathbf{V}_u(\lambda)^{H}\right)\le P_{\mathrm{BS}},\quad
\lambda\!\left(\sum_{u=1}^{U}\mathrm{Tr}\!\left(\mathbf{V}_u(\lambda)\mathbf{V}_u(\lambda)^{H}\right)-P_{\mathrm{BS}}\right)=0.
\end{aligned}
\label{eq:Vproof_kkt}
\end{equation}
\noindent\rule{\textwidth}{0.4pt}
\end{figure*}
Finally, define $P(\lambda)\triangleq \sum_{u=1}^{U}\mathrm{Tr}\big(\mathbf{V}_u(\lambda)\mathbf{V}_u(\lambda)^{H}\big)$.
Since $(\mathbf{A}_u+\lambda\mathbf{I}_L)^{-1}$ decreases in PSD order as $\lambda$ increases, $P(\lambda)$ is monotonically decreasing,
so $\lambda$ can be found efficiently via bisection whenever the power constraint is active (otherwise $\lambda=0$).
\end{proof}

\subsection{Holographic Amplitude Subproblem for $\mathbf{m}$}
\label{subsec:m_update}

Fix $\{\mathbf{V}_u\}$ and $\{g_{k,u},w_{k,u}\}$. We adopt a coupling freezing step by treating
$\mathbf{C}_u$ as a constant evaluated at the current iterate, so that the coupling aware RHS operator becomes affine in $\mathbf{m}$:
$\mathbf{M}_u(\mathbf{m})=\mathbf{C}_u\,\mathbf{D}(\mathbf{m})\,\mathbf{F}_u$.

\begin{theorem}[Convex quadratic form of the $\mathbf{m}$ subproblem]
\label{thm:m_qcqp_unified}
Under $\mathbf{M}_u(\mathbf{m})=\mathbf{C}_u\mathbf{D}(\mathbf{m})\mathbf{F}_u$ with fixed $\mathbf{C}_u$,
the WMMSE update of $\mathbf{m}\in\mathbb{R}^N$ reduces to a convex quadratic program with box constraints,
and becomes a convex quadratically constrained quadratic program when the RHS power constraint is included:
\begin{equation}
\begin{aligned}
\min_{\mathbf{m}\in\mathbb{R}^N}\quad &
\mathbf{m}^{T}\mathbf{Q}\mathbf{m}-2\Re\!\left\{\mathbf{q}^{H}\mathbf{m}\right\} \\
\text{s.t.}\quad &
\mathbf{0}\le \mathbf{m}\le \mathbf{1},\\
&
\eta\sum_{u=1}^{U}\mathbf{m}^{T}\mathbf{R}_u\mathbf{m}\le P_{\mathrm{RHS}},
\end{aligned}
\label{eq:m_qcqp_unified}
\end{equation}
where $\mathbf{Q}\succeq \mathbf{0}$ and $\mathbf{R}_u\succeq \mathbf{0}$ for all $u$.
\end{theorem}

\begin{proof}
We first express the coupling-frozen effective scalar response as a linear function of the real hologram vector $\mathbf{m}\in\mathbb{R}^N$.
Define $\mathbf{f}_{i,u}\triangleq \mathbf{F}_u\mathbf{v}_{i,u}\in\mathbb{C}^{N\times 1}$ and
$\mathbf{r}_{k,u}\triangleq \mathbf{h}_{k,u}\mathbf{C}_u\in\mathbb{C}^{1\times N}$. Then
$
z_{k,i,u}(\mathbf{m}) \triangleq \mathbf{h}_{k,u}\mathbf{C}_u\mathbf{D}(\mathbf{m})\mathbf{f}_{i,u}.
$
By expanding the diagonal modulation, $z_{k,i,u}(\mathbf{m})$ becomes
$
z_{k,i,u}(\mathbf{m})
=
\sum_{n=1}^{N}\mathbf{r}_{k,u}[n]\;m_n\;\mathbf{f}_{i,u}[n].
$
Let
$
\mathbf{a}_{k,i,u}\triangleq \left(\mathbf{r}_{k,u}^{T}\odot \mathbf{f}_{i,u}\right)^{*}\in\mathbb{C}^{N\times 1},
$
which yields the key linear relation
$
z_{k,i,u}(\mathbf{m})=\mathbf{a}_{k,i,u}^{H}\mathbf{m}.
$
Since $\mathbf{m}$ is real, the corresponding magnitude square can be written as a real quadratic form
$
\left|z_{k,i,u}(\mathbf{m})\right|^{2}
=
\left|\mathbf{a}_{k,i,u}^{H}\mathbf{m}\right|^{2}
=
\mathbf{m}^{T}\Re\!\left(\mathbf{a}_{k,i,u}\mathbf{a}_{k,i,u}^{H}\right)\mathbf{m}.
$

Substituting $z_{k,i,u}(\mathbf{m})$ into the MSE under the fixed equalizer $g_{k,u}$ gives
$
e_{k,u}(\mathbf{m})
=
|g_{k,u}|^2\Big(\sum_{i=1}^{K}|z_{k,i,u}(\mathbf{m})|^2+\sigma_{k,u}^2\Big) -2\Re\!\left\{g_{k,u}\,z_{k,k,u}(\mathbf{m})\right\}+1.
$
Therefore $e_{k,u}(\mathbf{m})$ is a quadratic function of $\mathbf{m}$:
$
e_{k,u}(\mathbf{m})
=
\mathbf{m}^{T}\mathbf{Q}_{k,u}\mathbf{m}-2\Re\!\left\{\mathbf{q}_{k,u}^{H}\mathbf{m}\right\}+c_{k,u},
$
where
$
\mathbf{Q}_{k,u}\triangleq |g_{k,u}|^2\sum_{i=1}^{K}\Re\!\left(\mathbf{a}_{k,i,u}\mathbf{a}_{k,i,u}^{H}\right)\succeq \mathbf{0}$,
$
\mathbf{q}_{k,u}\triangleq g_{k,u}^{*}\mathbf{a}_{k,k,u}$,
$
c_{k,u}\triangleq |g_{k,u}|^2\sigma_{k,u}^2+1$.
Finally, the weighted sum MSE part of the WMMSE objective can be collected as
$
\sum_{u=1}^{U}\sum_{k=1}^{K} w_{k,u}e_{k,u}(\mathbf{m})
=
\mathbf{m}^{T}\mathbf{Q}\mathbf{m}-2\Re\!\left\{\mathbf{q}^{H}\mathbf{m}\right\}+\text{const},
$
with the aggregated coefficients
$\mathbf{Q}\triangleq \sum_{u=1}^{U}\sum_{k=1}^{K} w_{k,u}\mathbf{Q}_{k,u}\succeq \mathbf{0}$,
$\mathbf{q}\triangleq \sum_{u=1}^{U}\sum_{k=1}^{K} w_{k,u}\mathbf{q}_{k,u}$.

Next, consider the RHS power term. Let $\mathbf{S}_u\triangleq \mathbf{F}_u\mathbf{V}_u\mathbf{V}_u^{H}\mathbf{F}_u^{H}\succeq \mathbf{0}$
and $\mathbf{G}_u\triangleq \mathbf{C}_u^{H}\mathbf{C}_u\succeq \mathbf{0}$. Then
\begin{equation}
\begin{aligned}
&\mathrm{Tr}\!\left(\mathbf{M}_u(\mathbf{m})\mathbf{V}_u\mathbf{V}_u^{H}\mathbf{M}_u(\mathbf{m})^{H}\right) \\
=&\left\|\mathbf{C}_u\mathbf{D}(\mathbf{m})\mathbf{F}_u\mathbf{V}_u\right\|_{F}^{2}\\
=&\mathrm{Tr}\!\left(\mathbf{D}(\mathbf{m})\mathbf{G}_u\mathbf{D}(\mathbf{m})\mathbf{S}_u\right)\\
=&\sum_{n=1}^{N}\sum_{n'=1}^{N} m_nm_{n'}\,[\mathbf{G}_u]_{n,n'}\,[\mathbf{S}_u^{T}]_{n,n'}\\
=&\mathbf{m}^{T}\Re\!\left(\mathbf{G}_u\odot \mathbf{S}_u^{T}\right)\mathbf{m}
\triangleq \mathbf{m}^{T}\mathbf{R}_u\mathbf{m},
\end{aligned}
\label{eq:m_proof_power_long}
\end{equation}
where $\mathbf{R}_u\triangleq \Re\!\left(\mathbf{G}_u\odot \mathbf{S}_u^{T}\right)\succeq \mathbf{0}$ by the Schur product theorem.
Therefore, the objective in \eqref{eq:m_qcqp_unified} is convex quadratic in $\mathbf{m}$, the box constraint is convex,
and the RHS power constraint is a convex quadratic constraint. This completes the proof.
\end{proof}

In practice, \eqref{eq:m_qcqp_unified} can be solved by a standard convex solver. For large $N$, a first order method is convenient.
Let $f(\mathbf{m})=\mathbf{m}^{T}\mathbf{Q}\mathbf{m}-2\Re\{\mathbf{q}^{H}\mathbf{m}\}$. A projected gradient step reads
$
\mathbf{m}^{(t+1)}=\Pi_{[0,1]^N}\!\Big(\mathbf{m}^{(t)}-\rho^{(t)}\nabla f(\mathbf{m}^{(t)})\Big),
$
$
\nabla f(\mathbf{m})=2\Re(\mathbf{Q})\mathbf{m}-2\Re(\mathbf{q}),
$
where $\Pi_{[0,1]^N}$ denotes elementwise clipping to $[0,1]$. When the RHS power constraint is active, the iterate can be further projected onto
$\{\mathbf{m}:\eta\sum_u \mathbf{m}^T\mathbf{R}_u\mathbf{m}\le P_{\mathrm{RHS}}\}$.

\begin{theorem}[Monotonic decrease and convergence of WMMSE block coordinate descent (BCD)]
\label{thm:wmmse_bcd_convergence_unified}
Let $J(\mathbf{m},\{\mathbf{V}_u\},\{g_{k,u},w_{k,u}\})=\sum_{u,k}\big(w_{k,u}e_{k,u}-\ln w_{k,u}\big)$.
At iteration $t$, update $\{g_{k,u}\}$ by the MMSE rule, update $\{w_{k,u}\}$ by $w_{k,u}=1/e_{k,u}$, update $\{\mathbf{V}_u\}$ by solving
the convex $\{\mathbf{V}_u\}$ subproblem, and update $\mathbf{m}$ by solving \eqref{eq:m_qcqp_unified}.
Then $\{J^{(t)}\}$ is monotonically nonincreasing and convergent. Any accumulation point of $\big(\mathbf{m}^{(t)},\{\mathbf{V}_u^{(t)}\}\big)$
is a stationary point of the WMMSE formulation, hence corresponds to a stationary point of the original sum rate problem.
\end{theorem}

\begin{proof}
Each block update minimizes $J$ with respect to one variable block while keeping the other blocks fixed.
Therefore the objective value cannot increase after any single block update, and a telescoping inequality chain follows.
\begin{equation}
\begin{aligned}
J^{(t+1)}
&= J\!\left(\mathbf{m}^{(t+1)},\{\mathbf{V}_u^{(t+1)}\},\{g_{k,u}^{(t+1)},w_{k,u}^{(t+1)}\}\right)\\
&\le J\!\left(\mathbf{m}^{(t)},\{\mathbf{V}_u^{(t+1)}\},\{g_{k,u}^{(t+1)},w_{k,u}^{(t+1)}\}\right)\\
&\le J\!\left(\mathbf{m}^{(t)},\{\mathbf{V}_u^{(t)}\},\{g_{k,u}^{(t+1)},w_{k,u}^{(t+1)}\}\right)\\
&\le J\!\left(\mathbf{m}^{(t)},\{\mathbf{V}_u^{(t)}\},\{g_{k,u}^{(t)},w_{k,u}^{(t+1)}\}\right)\\
&\le J\!\left(\mathbf{m}^{(t)},\{\mathbf{V}_u^{(t)}\},\{g_{k,u}^{(t)},w_{k,u}^{(t)}\}\right)
= J^{(t)}.
\end{aligned}
\label{eq:convergence_chain_long}
\end{equation}

Since $\sigma_{k,u}^2>0$, the MMSE satisfies $e_{k,u}>0$, and the function $w_{k,u}e_{k,u}-\ln w_{k,u}$ is bounded below over $w_{k,u}>0$.
Hence $J^{(t)}$ is lower bounded and the monotone sequence converges.
Standard block coordinate descent arguments then imply that any accumulation point satisfies blockwise optimality, which yields stationarity.
\end{proof}

\begin{algorithm}[t]
\caption{WMMSE BCD for joint $\{\mathbf{V}_u\}$ and amplitude hologram $\mathbf{m}$}
\label{alg:wmmse_bcd_rhs}
\begin{algorithmic}[1]
\REQUIRE $\{\mathbf{h}_{k,u}\}$, $\{\mathbf{F}_u,\boldsymbol{\Xi}_u\}$, $\{\sigma_{k,u}^2\}$, $P_{\mathrm{BS}}$, box constraint $0\le \mathbf{m}\le \mathbf{1}$, optional $P_{\mathrm{RHS}}$.
\ENSURE $\mathbf{m}$, $\{\mathbf{V}_u\}$.
\STATE Initialize $\mathbf{m}^{(0)}\in[0,1]^N$ and $\{\mathbf{V}_u^{(0)}\}$ with $\sum_u \mathrm{Tr}(\mathbf{V}_u^{(0)}\mathbf{V}_u^{(0)H})\le P_{\mathrm{BS}}$.
\FOR{$t=0,1,\ldots,T-1$}
\STATE Evaluate $\mathbf{M}_u(\mathbf{m}^{(t)})=\big(\mathbf{I}-\mathbf{D}(\mathbf{m}^{(t)})\boldsymbol{\Xi}_u\big)^{-1}\mathbf{D}(\mathbf{m}^{(t)})\mathbf{F}_u$ and $\overline{\mathbf{h}}_{k,u}^{(t)}=\mathbf{h}_{k,u}\mathbf{M}_u(\mathbf{m}^{(t)})$.
\STATE Update $\{g_{k,u}^{(t+1)}\}$ and $\{w_{k,u}^{(t+1)}\}$ using the MMSE closed forms.
\STATE Update $\{\mathbf{V}_u^{(t+1)}\}$ by the KKT closed form with a bisection search on $\lambda^{(t+1)}$ to satisfy $\sum_u \mathrm{Tr}(\mathbf{V}_u\mathbf{V}_u^{H})\le P_{\mathrm{BS}}$.
\STATE Jacobian-aided update $\mathbf{m}^{(t+1)}$ by solving \eqref{eq:m_qcqp_unified} using projected gradient and projection onto the feasible set.
\ENDFOR
\STATE Output $\mathbf{m}^{(T)}$ and $\{\mathbf{V}_u^{(T)}\}$.
\end{algorithmic}
\end{algorithm}

\subsection{Jacobian-Aided Coupling-Consistent Hologram Update}
\label{subsec:jacobian_m_update}
Since the WSR--WMMSE equivalence is standard, we focus on a coupling-consistent hologram update
that accounts for the dependence of
$
\mathbf{C}_u(\mathbf{m})\triangleq(\mathbf{I}-\mathbf{D}(\mathbf{m})\boldsymbol{\Xi}_u)^{-1}
$
on the hologram amplitude \(\mathbf{m}\).
The coupling-freeze step in Sec.~\ref{subsec:m_update} uses
$\mathbf{M}_u(\mathbf{m})\approx \mathbf{C}_u^{(t)}\mathbf{D}(\mathbf{m})\mathbf{F}_u$,
which ignores \(\partial \mathbf{C}_u/\partial \mathbf{m}\) and may degrade under strong mutual coupling.
We instead adopt a first-order coupling-consistent surrogate (derivations in Appendix~\ref{app:jacobian_proof}).

\subsubsection{Coupling-consistent first-order surrogate}
Recall
$
\mathbf{M}_u(\mathbf{m})=\mathbf{C}_u(\mathbf{m})\mathbf{D}(\mathbf{m})\mathbf{F}_u
$
with
$
\mathbf{C}_u(\mathbf{m})=(\mathbf{I}-\mathbf{D}(\mathbf{m})\boldsymbol{\Xi}_u)^{-1}.
$
By implicit differentiation, the variation of \(\mathbf{M}_u\) obeys
$
\mathrm{d}\mathbf{M}_u
=
\mathbf{C}_u(\mathrm{d}\mathbf{D})
\Big(\boldsymbol{\Xi}_u\mathbf{M}_u+\mathbf{F}_u\Big),
$
which captures not only the direct scaling via \(\mathbf{D}(\mathbf{m})\) but also the coupling-induced feedback term
\(\boldsymbol{\Xi}_u\mathbf{M}_u\).
At iteration \(t\), let \(\mathbf{m}^{(t)}\), \(\mathbf{C}_u^{(t)}=\mathbf{C}_u(\mathbf{m}^{(t)})\),
\(\mathbf{M}_u^{(t)}=\mathbf{M}_u(\mathbf{m}^{(t)})\), and \(\Delta\mathbf{m}\triangleq \mathbf{m}-\mathbf{m}^{(t)}\).
Define
$
\mathbf{T}_u^{(t)} \triangleq \boldsymbol{\Xi}_u\mathbf{M}_u^{(t)}+\mathbf{F}_u,
$
$
\widetilde{\mathbf{M}}_u(\mathbf{m}\,|\,\mathbf{m}^{(t)})
\triangleq
\mathbf{M}_u^{(t)}+\mathbf{C}_u^{(t)}\mathbf{D}(\Delta\mathbf{m})\mathbf{T}_u^{(t)} .
$

\subsubsection{Convex quadratic WMSE and QCQP update}
Using \(\widetilde{\mathbf{M}}_u\), define the surrogate scalar response
$
z_{k,i,u}(\mathbf{m})\triangleq
\mathbf{h}_{k,u}\widetilde{\mathbf{M}}_u(\mathbf{m}|\mathbf{m}^{(t)})\mathbf{v}_{i,u}
$
and \(z_{k,i,u}^{(t)}\triangleq \mathbf{h}_{k,u}\mathbf{M}_u^{(t)}\mathbf{v}_{i,u}\).
Let
$
\mathbf{r}_{k,u}^{(t)}\triangleq \mathbf{h}_{k,u}\mathbf{C}_u^{(t)}\in\mathbb{C}^{1\times N}
$,
$
\mathbf{f}_{i,u}^{(t)}\triangleq \mathbf{T}_u^{(t)}\mathbf{v}_{i,u}\in\mathbb{C}^{N\times 1}
$,
and
$
\mathbf{a}_{k,i,u}^{(t)}\triangleq
\left(\big(\mathbf{r}_{k,u}^{(t)}\big)^{T}\odot \mathbf{f}_{i,u}^{(t)}\right)^{*},
$
$
z_{k,i,u}(\mathbf{m})
=
z_{k,i,u}^{(t)}+\big(\mathbf{a}_{k,i,u}^{(t)}\big)^{H}\Delta\mathbf{m}.
$
Substituting it into the WMSE yields a convex quadratic function of
\(\Delta\mathbf{m}\in\mathbb{R}^{N}\) (details in Appendix~\ref{app:jacobian_proof}):
$
\sum_{u,k} w_{k,u} e_{k,u}
=
\Delta\mathbf{m}^{T}\mathbf{Q}^{(t)}\Delta\mathbf{m}
-2\Re\!\left\{\big(\mathbf{q}^{(t)}\big)^{H}\Delta\mathbf{m}\right\}
+\mathrm{const},
$
where
\begin{align}
\mathbf{Q}^{(t)} &\triangleq
\sum_{u,k} w_{k,u}|g_{k,u}|^{2}\sum_{i=1}^{K}
\Re\!\left(\mathbf{a}_{k,i,u}^{(t)}\big(\mathbf{a}_{k,i,u}^{(t)}\big)^{H}\right)\succeq \mathbf{0},
\\
\mathbf{q}^{(t)} &\triangleq
\sum_{u,k} w_{k,u}\left(
g_{k,u}^{*}\mathbf{a}_{k,k,u}^{(t)}
-|g_{k,u}|^{2}\sum_{i=1}^{K} z_{k,i,u}^{(t)}\,\mathbf{a}_{k,i,u}^{(t)}
\right).
\end{align}
Moreover, under the surrogate \(\widetilde{\mathbf{M}}_u\), the RHS transmit-power term admits a convex
quadratic form in \(\Delta\mathbf{m}\) (Appendix~\ref{app:jacobian_proof}):
\begin{equation}
\eta\sum_{u=1}^{U}\Big(
\Delta\mathbf{m}^{T}\mathbf{R}_u^{(t)}\Delta\mathbf{m}
+2\,\mathbf{b}_u^{(t)T}\Delta\mathbf{m}
+c_u^{(t)}
\Big)
\le P_{\mathrm{RHS}}.
\label{eq:jac_rhs_power_sur}
\end{equation}
Hence, the hologram update reduces to the following convex QCQP:
\begin{equation}
\begin{aligned}
\min_{\Delta\mathbf{m}\in\mathbb{R}^{N}}~~&
\Delta\mathbf{m}^{T}\mathbf{Q}^{(t)}\Delta\mathbf{m}
-2\Re\!\left\{\big(\mathbf{q}^{(t)}\big)^{H}\Delta\mathbf{m}\right\} \\
\mathrm{s.t.}~~&
-\mathbf{m}^{(t)}\le \Delta\mathbf{m}\le \mathbf{1}-\mathbf{m}^{(t)},\\
&\eqref{eq:jac_rhs_power_sur}.
\end{aligned}
\label{eq:jac_qcqp}
\end{equation}
The problem \eqref{eq:jac_qcqp} can be efficiently solved by projected gradient with the same order
complexity as Sec.~\ref{subsec:m_update}. We denote the resulting joint design as CA-Joint-Jac.

The overall WMMSE based block coordinate procedure is summarized in Algorithm~\ref{alg:wmmse_bcd_rhs}.

\section{Simulation}
\subsection{Simulation setup}
\subsubsection{Verification for the proposed coupling model}
We validate the proposed mutual coupling model against full-wave results using Meep \footnote{Meep: a free and open-source software package for simulating electromagnetic systems, available at: https://github.com/liyycat/meep.git} with a parallel plate waveguide (PPW) testbed. The length unit is set to 1 mm, and the operating frequency is $f_0$=27.1 GHz (free space wavelength $\lambda_0 =$ 11.06 mm) excited by a narrowband Gaussian source with fractional bandwidth $0.03 f_0$. The simulation cell size is $30 \times 18 \times 22.5$ mm$^3$, terminated by 2 mm perfectly matched layer (PML) on all boundaries. The parallel plate waveguide (PPW) is formed by two conductive plates of thickness 0.25 mm separated by 2.0 mm; the plates are modeled using Meep's conductivity-based metal approximation rather than an ideal perfect electric conductor (PEC). On the top plate, we etch five ``ON-state'' radiating apertures as lightweight equivalents of practical complementary Electric–LC resonator (cELC) RHS elements: each aperture is an air slot of size 0.5 $\times$ 1.2 mm$^2$, repeated with pitch 2.0 mm and centered on the array axis. A vertical electric dipole (z-directed electric-field component (Ez) source) is placed at the PPW mid-plane (0,0,0) to emulate the in-waveguide excitation. The spatial resolution is set to 25 pixels/mm (i.e., spacing 0.04 mm), and Meep’s subpixel averaging is enabled by default to reduce staircase errors at metal air boundaries. A diagram is given by Fig. \ref{fig:diagram}.
\begin{figure}[h]
\centering
\includegraphics[width=0.5\textwidth]{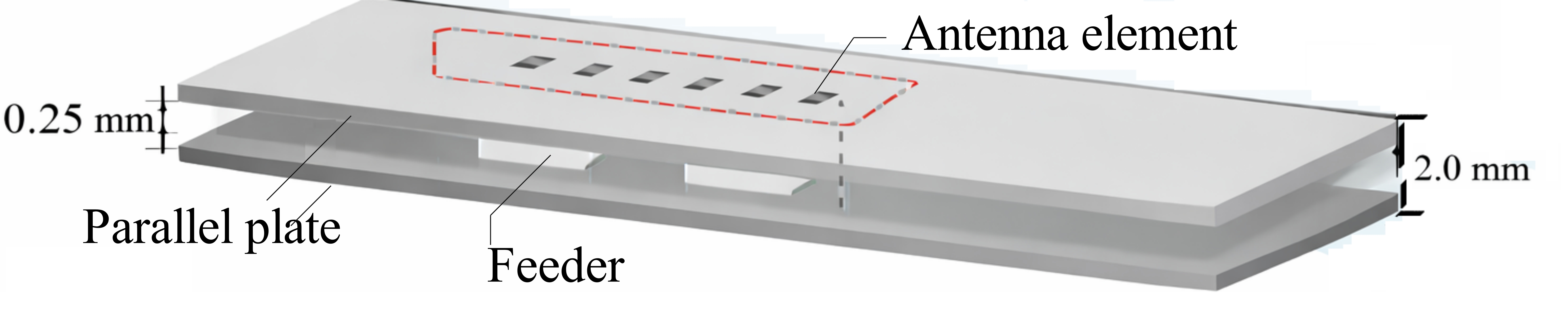}
\caption{Parallel-plate waveguide (PPW) testbed used in Meep for validating the mutual-coupling model: five subwavelength radiating apertures (slot proxies of ON-state cELC RHS elements) are etched on the top conductive plate (thickness 0.25 mm) above a 2.0 mm plate separation.}
\label{fig:diagram}
\end{figure}

This Meep testbed mimics RHS by using a parallel-plate waveguide as the traveling reference-wave feeder and subwavelength top-plate slots as ON-state leaky radiators. Full-wave simulation naturally includes multi-scattering mutual coupling, enabling beampattern validation, though circuit-level PIN/resonance details are abstracted.

\begin{table}[t]
\centering
\caption{Simulation parameters for beamforming.}
\label{tab:sim_params}
\begin{tabular}{|c c|}
\hline
Parameter & Value \\
\hline
Carrier frequency $f_c$ & 28 GHz \\
Total bandwidth $B$ & 1 GHz \\
Number of subbands $U$ & 8 \\
Subband bandwidth $B_g$ & 125 MHz \\
Subband center frequencies $f_u$ (GHz) & $\{$27.56, 27.68, $\ldots$, 28.31, 28.43$\}$ \\
Medium permeability $\mu$ (H/m) & $4\pi\times10^{-7}$ \\
Medium permittivity $\epsilon$ (F/m) & $8.854187817\times10^{-12}$ \\
Wave speed $c_0$ (m/s) & $2.99792458\times10^{8}$ \\
Array type & ULA \\
Number of RHS elements $N$ & 32 \\
Element spacing $d$ & 2.68 mm \\
Number of feeders $L$ & 4 \\
Feeder spacing & 10.70 mm \\
Dipole orientation $\hat{\mathbf{e}}_{\mathrm{m}}$ & $[0,0,1]^T$ \\
Number of users $K$ & 4 \\
User distances $r_k$ (m) & [3.0, 4.5, 6.0, 7.5] \\
User angles $\theta_k$ (deg) & [75, 85, 95, 105] \\
Absorption coefficient $\kappa_{\mathrm{abs}}(f_u)$ & 0.1 \\
BS power budget $P_{\mathrm{BS}}$ & \{2, 5, 10, 20\} \\
RHS power budget $P_{\mathrm{RHS}}$, efficiency $\eta$ & 50, 1 \\
Noise power $\sigma^2$ & 1 \\
Waveguide coupling strength $\xi_{\mathrm{wg}}$ & 0.02\\
Waveguide attenuation $\alpha_{\mathrm{wg}}$ & 0.15\\
Waveguide factor $\beta_{\mathrm{wg}}$ & 1.0\\
Stop threshold & 1e-4 \\
Max iterations $T$ & 100 \\
step size & 0.05 \\
\hline
\end{tabular}
\end{table}

\subsubsection{Verification for the proposed coupling-aware RHS beamforming for multi-user sum rate maximization}
The simulation setup follows Table~\ref{tab:sim_params}. For each Monte Carlo run, we generate a wideband RHS MU-MIMO scenario with $U$ subbands, $N$ RHS elements, $L$ feeders, and $K$ singleantenna users. The true mutual coupling matrix is constructed by the parametric coupling model, and all schemes are evaluated under the same true coupling and channel realizations for fair comparison. The proposed algorithm adopts a WMMSE-BCD procedure that alternates among updating the MMSE receiver/weights, updating the digital precoders under the BS power constraint, and updating the RHS hologram vector via projected gradient descent (when applicable). Performance is measured by the achieved sum spectral efficiency (equivalently sum rate over bandwidth), while the RHS loaded power and the surrogate objective $J= \sum_{u, k}\left(1-\ln \left(1+\operatorname{SINR}_{k, u}\right)\right)$ are also recorded to illustrate convergence. We compare six baselines: 
\begin{itemize}
  \item {CA-Joint}: coupling-aware joint optimization of the RHS hologram $\mathbf{m}$ and the digital precoders $\{\mathbf{V}_u\}$ via WMMSE-BCD, using the true coupling model in the design.
  \item {CU-Joint}: coupling-unaware joint optimization of $\mathbf{m}$ and $\{\mathbf{V}_u\}$ assuming zero coupling during the design, while performance is evaluated under the true coupling.
  \item {CA-Joint-Jac}: coupling-aware joint optimization of the RHS hologram $\mathbf{m}$ and the digital precoders $\{\mathbf{V}_u\}$ via WMMSE-BCD, with Jacobian-aided surrogate \eqref{eq:M_tilde} instead of frozening $C_u$.
  \item {Holo+WMMSE}: fixed amplitude-only hologram $\mathbf{m}$ (holographic pattern), and only $\{\mathbf{V}_u\}$ is optimized by WMMSE under the true coupling.
  \item {Uniform+WMMSE}: fixed uniform $\mathbf{m}$, and only $\{\mathbf{V}_u\}$ is optimized by WMMSE under the true coupling.
  \item {Holo+ZF}: fixed hologram $\mathbf{m}$, and $\{\mathbf{V}_u\}$ is designed by subband-wise zero-forcing (ZF) precoding under the true coupling.
  \item {Uniform+ZF}: fixed uniform $\mathbf{m}$, and $\{\mathbf{V}_u\}$ is designed by subband-wise ZF precoding under the true coupling.
\end{itemize}

\subsection{Simulation Results}
\label{sec:sim_results}
\subsubsection{Coupling model verification}

\begin{figure}[h]
\centering
\includegraphics[width=0.5\textwidth]{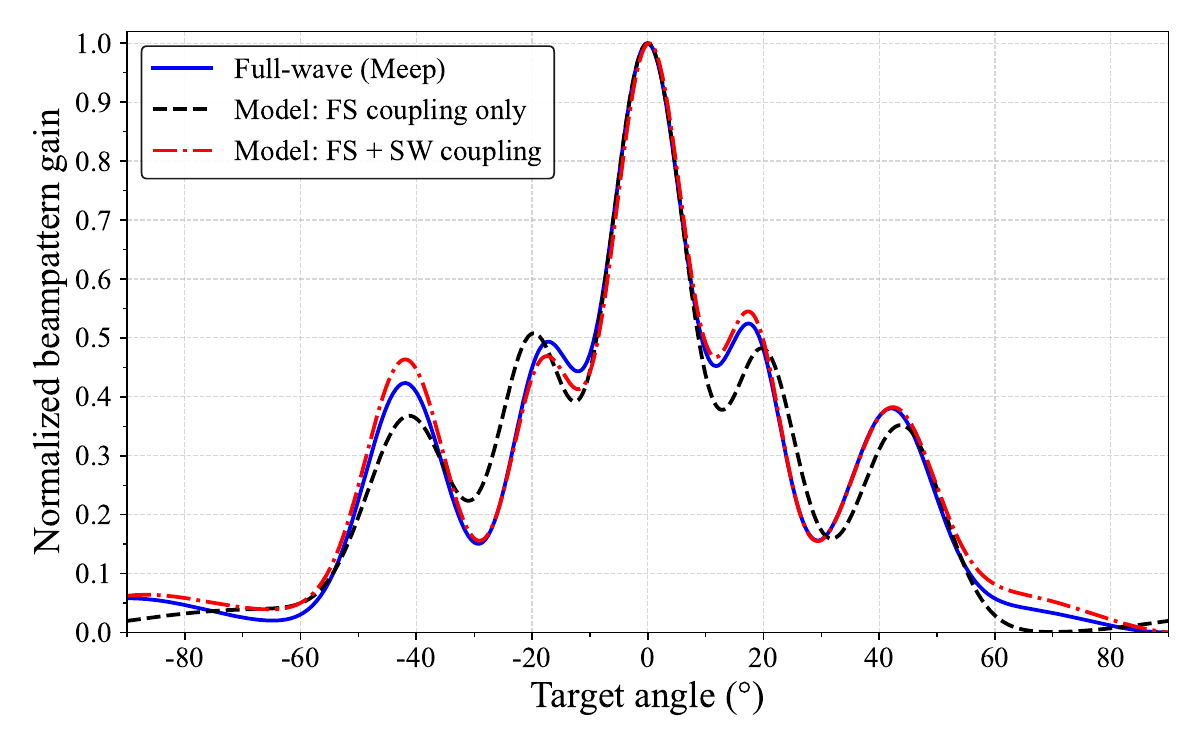}
\caption{Beampattern validation on the $90^{\circ}$ cut: full-wave (Meep) versus coupling-aware models (FS-only and FS+SW) for a 5-element cELC-based RHS on a parallel-plate waveguide.}
\label{fig:coupling_verf}
\end{figure}

\begin{figure*}[h]
\centering
\includegraphics[width=\textwidth]{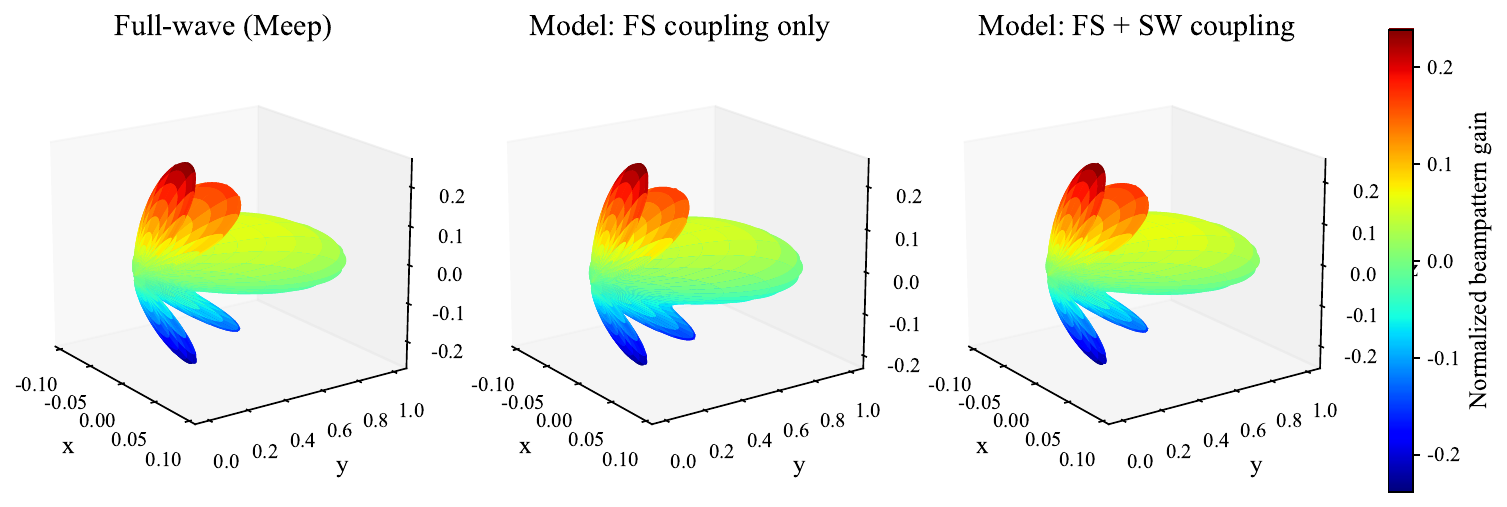}
\caption{3D normalized beampattern comparison: full-wave (Meep) and two coupled-dipole predictions (FS-only and FS+SW) under the same waveguidefed 5-element RHS configuration.}
\label{fig:coupling_verf_3d}
\end{figure*}

Fig. \ref{fig:coupling_verf} and \ref{fig:coupling_verf_3d} follows the designed validation setup (five ON-state cELC-like apertures etched on a parallel-plate waveguide, excited by a vertical in-waveguide dipole at 27.1 GHz, with the far-field extracted on the $90^\circ$ cut and normalized): the 2D cut confirms that the proposed FS+SW coupling model reproduces the full-wave main-lobe direction/width and the four dominant sidelobes, while the FS-only model mispredicts sidelobe heights and null depths; the corresponding 3D patterns further show that ignoring surface coupling distorts the off-boresight energy distribution, whereas including FS+SW coupling preserves the full-wave directional energy confinement and sidelobe geometry, validating the necessity of coupling-consistent modeling under dense spacing.

\subsubsection{Sum-rate convergence} Fig.~\ref{fig:sum_se_convergence} shows the sum rate $R$ versus the outer iteration index. The proposed joint designs (CA-Joint and CU-Joint) monotonically improve the achievable rate because each iteration updates the digital precoders via WMMSE and (for joint schemes) refines the hologram profile $\mathbf{m}$, thereby improving the effective channel quality and multiuser interference management. In contrast, the ZF baselines (Holo+ZF and Uniform+ZF) are essentially one-shot designs with fixed $\mathbf{m}$ and linear precoding; hence their curves remain almost constant across iterations. For the fixed-$\mathbf{m}$ WMMSE baselines, Uniform+WMMSE saturates quickly since the uniform hologram provides limited spatial degrees of freedom, and only the digital precoder can be optimized. Notably, Holo+WMMSE achieves the highest rate in this setting, indicating that the simple amplitude-only hologram provides a favorable initial effective channel, and optimizing only $\{\mathbf{V}_u\}$ is sufficient to exploit it. The joint schemes improve more gradually because they must balance between enhancing beamforming gain and maintaining feasibility under coupling and power constraints.
\begin{figure}[h]
\centering
\includegraphics[width=0.5\textwidth]{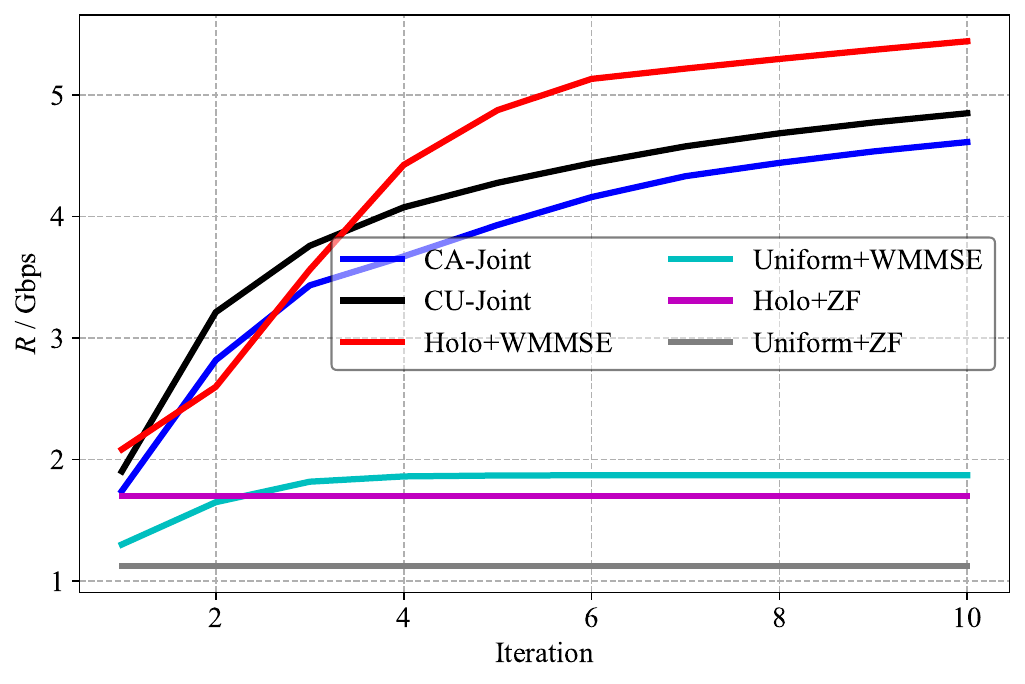}
\caption{Convergence curves of sum rate.}
\label{fig:sum_se_convergence}
\end{figure}

\subsubsection{Surrogate objective convergence}
Fig.~\ref{fig:surrogate_objective_convergence} plots the surrogate objective $J=\sum_{u,k}\big(1-\ln(1+\mathrm{SINR}_{k,u})\big)$, which is a monotone decreasing function of the true SINR. Therefore, a decreasing $J$ indicates improved link quality. Consistent with Fig.~X, the WMMSE-based schemes exhibit clear descent behavior, whereas the ZF baselines remain flat due to the absence of iterative refinement. Among the iterative methods, Holo+WMMSE yields the fastest decrease and the lowest final $J$, matching its superior sum-rate performance. This suggests that, with a well-chosen fixed hologram, digital WMMSE can effectively suppress multiuser interference without requiring hologram adaptation. The joint designs (CA-Joint and CU-Joint) also reduce $J$ but at a slower pace, since updating $\mathbf{m}$ under coupling introduces additional nonconvexity and model sensitivity; consequently, the algorithm trades off between improving SINR and preserving stability/feasibility of the coupled RHS response. Moreover, CU-Joint can appear slightly better than CA-Joint in terms of $J$ in finite iterations, which may occur when the coupling-aware hologram update becomes numerically more conservative (or locally trapped) due to the stronger nonlinearity induced by the true coupling model.
\begin{figure}[h]
\centering
\includegraphics[width=0.5\textwidth]{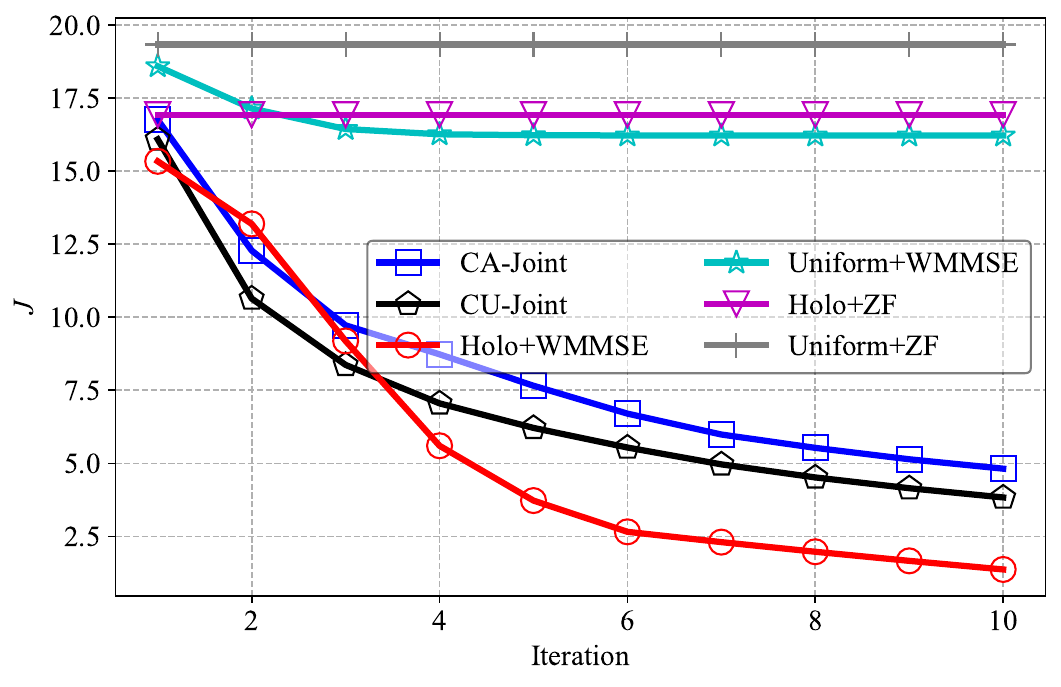}
\caption{Convergence curves of surrogate objective.}
\label{fig:surrogate_objective_convergence}
\end{figure}

\subsubsection{RHS loaded power evolution}
Fig.~\ref{fig:rhs_power_convergence} reports the RHS loaded power $P_{\mathrm{RHS}}$ over iterations. The joint schemes maintain $P_{\mathrm{RHS}}$ at a relatively low and stable level because the hologram update explicitly accounts for the RHS power constraint, and the digital precoder update is performed under a fixed BS power budget. In contrast, Holo+WMMSE exhibits a steadily increasing $P_{\mathrm{RHS}}$ as iterations proceed: although the BS transmit power is constrained, optimizing $\{\mathbf{V}_u\}$ for higher rate tends to concentrate energy into the effective RHS aperture (through $\mathbf{M}_u(\mathbf{m})$), thereby increasing the loaded power when $\mathbf{m}$ is fixed and not optimized for power efficiency. The ZF baselines remain nearly constant since both $\mathbf{m}$ and $\{\mathbf{V}_u\}$ are fixed after the one-shot design. These results highlight the importance of joint hologram--precoder optimization: adapting $\mathbf{m}$ enables controlling the RHS power consumption while still improving spectral efficiency, whereas V-only optimization may lead to excessive RHS loading even if the BS power is limited.

\begin{figure}[h]
\centering
\includegraphics[width=0.5\textwidth]{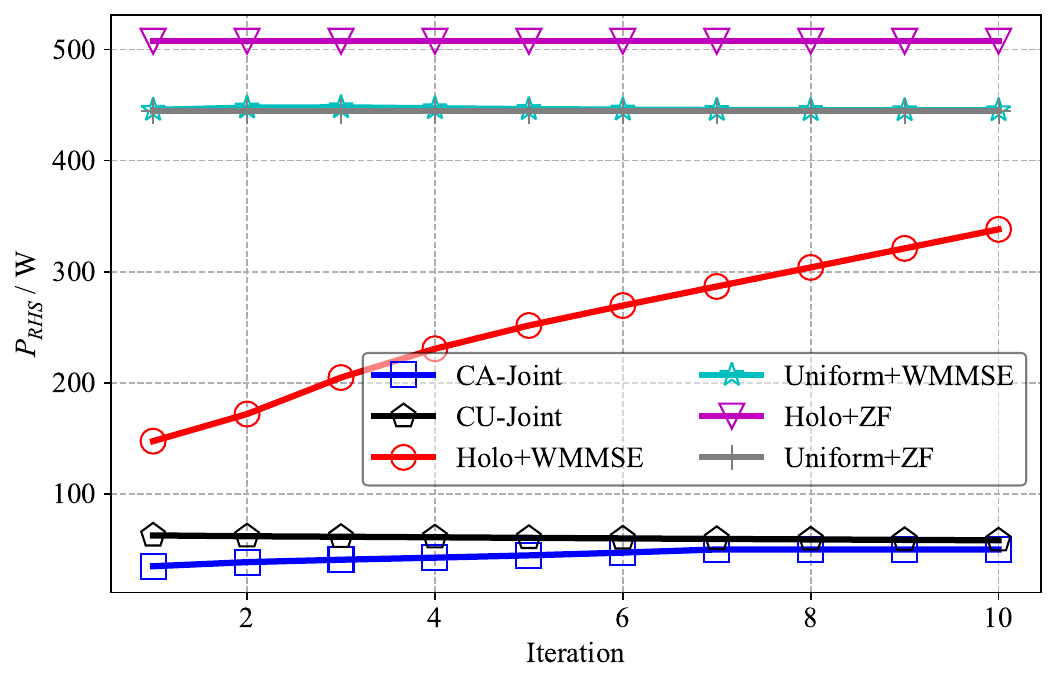}
\caption{Constraint tracking of RHS power.}
\label{fig:rhs_power_convergence}
\end{figure}

\subsubsection{Sum rate versus BS transmit power $P_{\mathrm{BS}}$}
Fig.~\ref{fig:sum_se_vs_pbs} depicts the achievable sum rate as a function of the BS power budget $P_{\mathrm{BS}}$. As expected, increasing $P_{\mathrm{BS}}$ generally improves the sum rate since the digital precoders can allocate more transmit energy to the effective RHS aperture. The fixed-$\mathbf{m}$ baselines exhibit distinct behaviors: {Uniform+WMMSE} increases steadily but remains significantly lower due to the limited beamforming capability of a uniform hologram, whereas {Holo+WMMSE} benefits much more from power scaling and achieves the highest rate at large $P_{\mathrm{BS}}$, indicating that the amplitude-only hologram provides a favorable effective channel that can be fully exploited by WMMSE precoding. In contrast, the joint designs \textit{CA-Joint} and {CU-Joint} show only marginal gains and even a slight decrease at high $P_{\mathrm{BS}}$, which suggests that the joint update of $\mathbf{m}$ (together with coupling and power-feasibility constraints) becomes the performance bottleneck: the algorithm tends to favor stable/feasible hologram profiles rather than aggressively increasing array gain, leading to saturation in the achievable rate within the considered iteration budget.

\begin{figure}[h]
\centering
\includegraphics[width=0.5\textwidth]{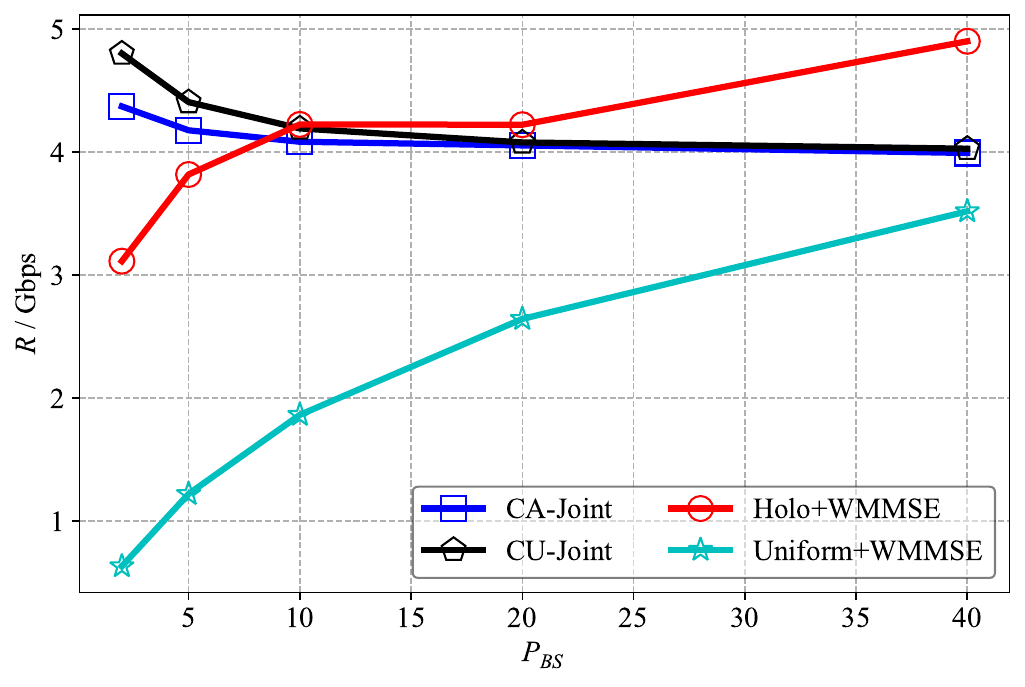}
\caption{Sum rate versus BS transmit power $P_\mathrm{BS}$.}
\label{fig:sum_se_vs_pbs}
\end{figure}

\subsubsection{Sum rate versus coupling strength $\xi_{\mathrm{fs}}$}
Fig.~\ref{fig:sum_se_vs_coupling} reports the sum rate under different near-field coupling strengths $\xi_{\mathrm{fs}}$. The \textit{Uniform+WMMSE} curve stays low and only slightly increases with $\xi_{\mathrm{fs}}$, reflecting that a uniform hologram cannot effectively leverage additional coupling-induced degrees of freedom. For \textit{Holo+WMMSE}, the sum rate is relatively insensitive to moderate coupling but decreases at strong coupling, because the fixed hologram is not optimized to mitigate coupling-induced distortion, and the digital precoder alone cannot fully compensate for the resulting mismatch in the effective channel. The joint schemes exhibit improved robustness: \textit{CA-Joint} remains stable across the sweep since the hologram update explicitly accounts for coupling in the design, while \textit{CU-Joint} may slightly outperform CA-Joint in this finite-iteration setting because ignoring coupling during optimization yields a smoother and less conservative update of $\mathbf{m}$; however, as coupling becomes stronger, the model mismatch in CU-Joint becomes more pronounced and may eventually limit its performance.

\begin{figure}[h]
\centering
\includegraphics[width=0.5\textwidth]{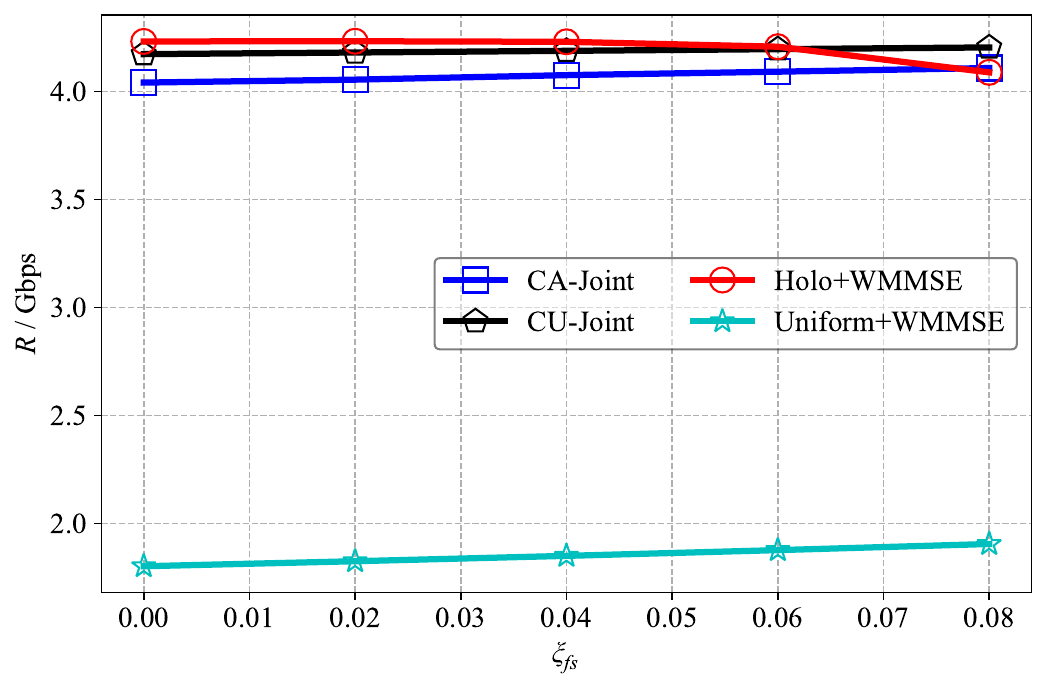}
\caption{Sum rate versus coupling strength.}
\label{fig:sum_se_vs_coupling}
\end{figure}

\subsubsection{Sum rate versus the number of RHS elements $N$}
Fig.~\ref{fig:sum_se_vs_rhs_size} shows the sum rate as a function of the RHS aperture size $N$. The joint designs \textit{CA-Joint} and \textit{CU-Joint} improve with $N$ because a larger RHS provides higher spatial resolution and more degrees of freedom for multiuser beamforming, enabling better interference suppression. The \textit{Uniform+WMMSE} baseline exhibits only modest gains since a uniform hologram does not exploit the additional aperture efficiently. In contrast, \textit{Holo+WMMSE} shows a strong dependence on $N$ and becomes the best-performing scheme for large apertures, demonstrating that a properly structured fixed hologram can already capture most of the array gain when combined with WMMSE digital precoding. The slight non-monotonic behavior of CA-Joint at large $N$ can be attributed to the increased nonconvexity and sensitivity of the hologram optimization in the presence of coupling, where the projected-gradient update may require more iterations or careful step-size tuning to fully benefit from the enlarged aperture.

\begin{figure}[h]
\centering
\includegraphics[width=0.5\textwidth]{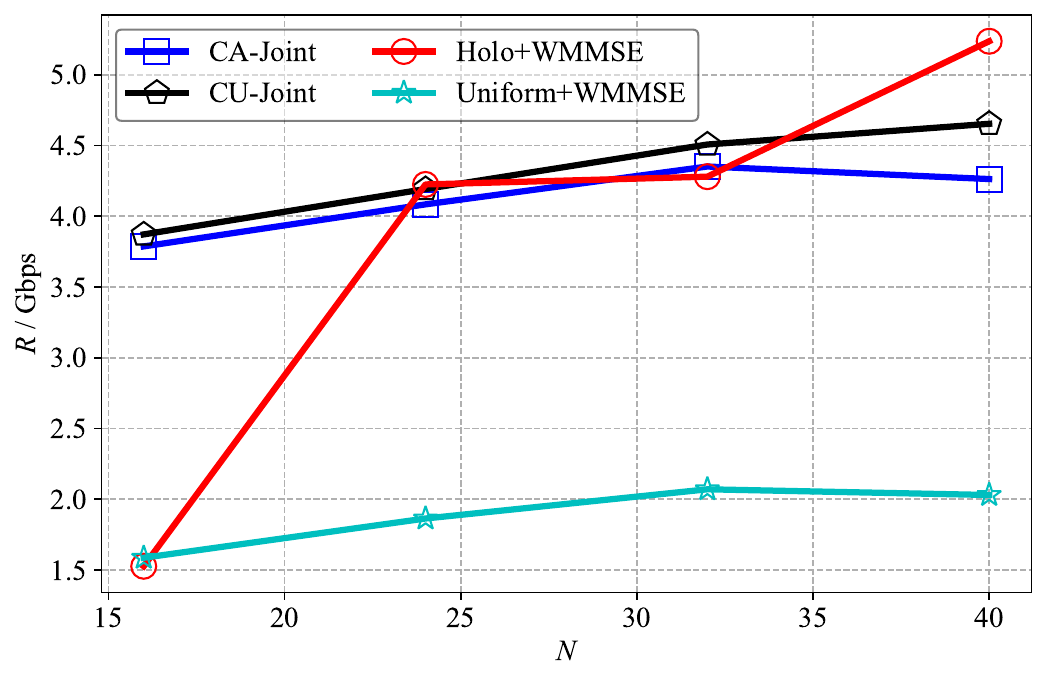}
\caption{Sum rate versus the number of RHS elements $N$.}
\label{fig:sum_se_vs_rhs_size}
\end{figure}

\begin{figure*}[t]
\centering
\includegraphics[width=0.48\textwidth]{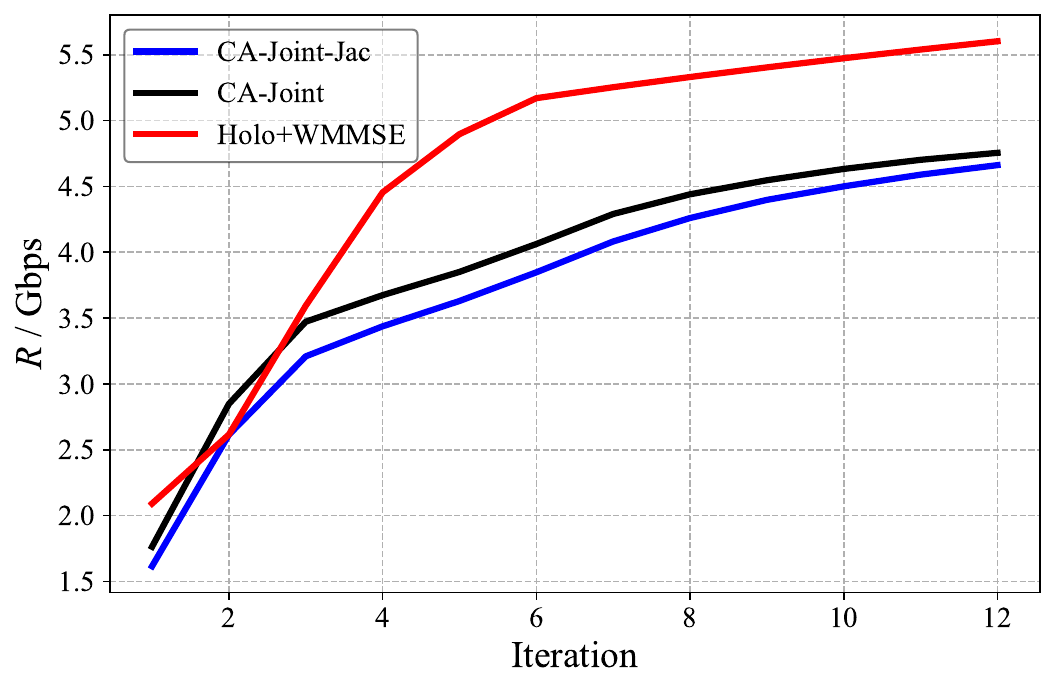}
\hfill
\includegraphics[width=0.48\textwidth]{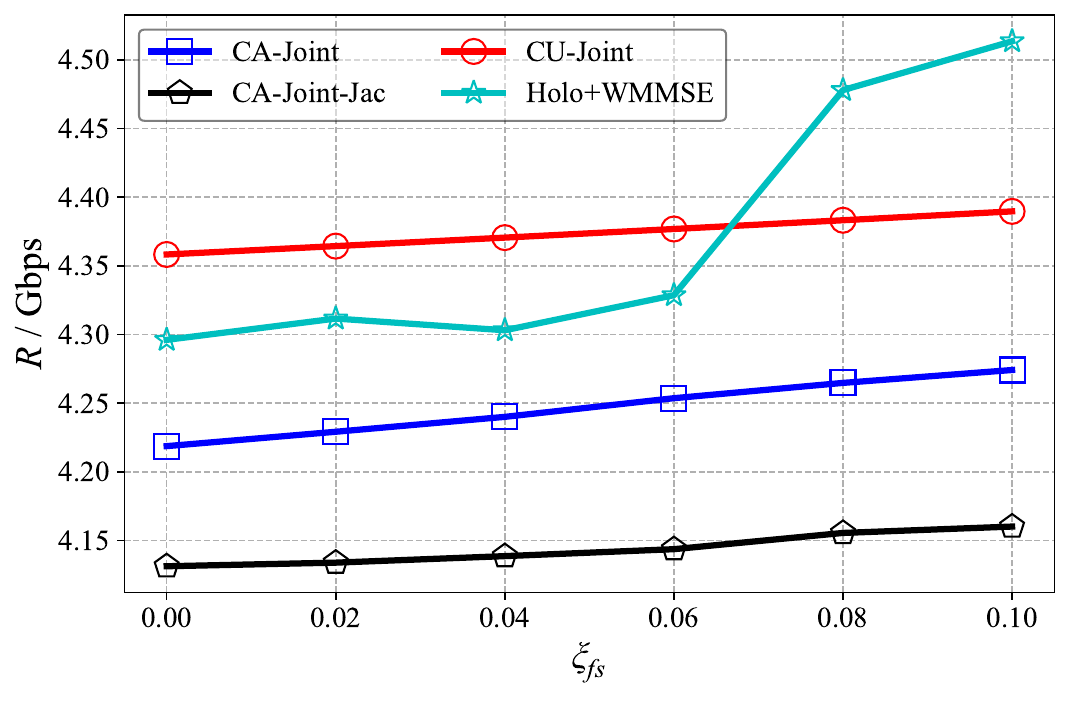}
\caption{Innovation validation of Jacobian-aided coupling-consistent hologram update.
Left: convergence under stronger coupling. Right: sum rate versus coupling strength \(\xi_{\mathrm{fs}}\).}
\label{fig:jac_innovation}
\end{figure*}

\subsubsection{Validating the Jacobian-aided update}
We add a baseline \textit{CA-Joint-Jac} that replaces the coupling-freeze hologram step by \eqref{eq:jac_qcqp}.
Fig.~\ref{fig:jac_innovation} (left) shows faster and higher-rate convergence under stronger coupling,
and Fig.~\ref{fig:jac_innovation} (right) shows improved robustness versus \(\xi_{\mathrm{fs}}\).

\section{Conclusion}
This work studies coupling-aware wideband MU-MIMO aided by RHSs and proposes a coupling-consistent joint hologram–precoder design. Modeling RHS elements as equivalent magnetic dipoles, we derive a frequency-dependent mutual-coupling matrix for subband-wise wideband evaluation and decompose it into free-space near-field and guided/surface-wave components for interpretability. We then maximize wideband sum rate over the holographic amplitude pattern and digital precoders under feeder sum-power and RHS excitation-power constraints. The resulting coupled nonconvex problem is solved via a WMMSE-BCD algorithm: digital precoders have a KKT closed-form update with a single bisection-searched dual variable. For the hologram, we propose a Jacobian-aided coupling-consistent update by implicitly differentiating the coupled RHS operator and building a first-order SCA surrogate that retains coupling-feedback sensitivity, yielding a convex quadratic/QCQP subproblem solvable by projected first-order methods. Simulations at 28 GHz with 1 GHz bandwidth validate fast convergence and improved robustness to coupling strength and aperture scaling. Future work includes measurement-driven coupling calibration, scalable large-aperture designs, and hardware-constrained holograms under imperfect CSI and multicell interference.

\bibliographystyle{Bibliography/IEEEtranTIE}
\bibliography{Bibliography/IEEEabrv,Bibliography/mybibfile.bib} %IEEEabrv instead of IEEEfull

\appendices

\section{Derivation of Eq.~\eqref{eq:H}}
\label{appA}
For $R_{n'n}\neq 0$, consider the standard Helmholtz kernel
$\frac{e^{-jk_u R_{n'n}}}{4\pi R_{n'n}}$ and apply the vector identity
$\nabla\times\nabla\times(\hat{\mathbf{e}}_{\mathrm{m}}\psi)
=\nabla(\hat{\mathbf{e}}_{\mathrm{m}}\cdot\nabla\psi)-\hat{\mathbf{e}}_{\mathrm{m}}\nabla^2\psi$
(valid since $\hat{\mathbf{e}}_{\mathrm{m}}$ is constant), with
$\psi=\frac{e^{-jk_u R_{n'n}}}{4\pi R_{n'n}}$.
Using $\nabla_{\mathbf{r}_{n'}} R_{n'n}=\widehat{\mathbf{R}}_{n'n}$,
$\nabla_{\mathbf{r}_{n'}}\widehat{\mathbf{R}}_{n'n}
=\frac{1}{R_{n'n}}\big(\mathbf{I}-\widehat{\mathbf{R}}_{n'n}\widehat{\mathbf{R}}_{n'n}^T\big)$,
and the scalar relations (for $R_{n'n}\neq 0$)
$
\nabla_{\mathbf{r}_{n'}}\psi
=\frac{e^{-jk_u R_{n'n}}}{4\pi}\Big(-\frac{1}{R_{n'n}^2}+\frac{jk_u}{R_{n'n}}\Big)\widehat{\mathbf{R}}_{n'n},
$
$
\nabla_{\mathbf{r}_{n'}}^2\psi=-k_u^2\psi,
$
one obtains after straightforward algebra
\begin{equation}
\begin{aligned}
&\nabla_{\mathbf{r}_{n'}}(\hat{\mathbf{e}}_{\mathrm{m}}\cdot\nabla_{\mathbf{r}_{n'}}\psi)
=
\frac{e^{-jk_u R_{n'n}}}{4\pi}
\Big(\frac{1}{R_{n'n}^2}-\frac{jk_u}{R_{n'n}}\Big)
\\
&\frac{3\widehat{\mathbf{R}}_{n'n}(\widehat{\mathbf{R}}_{n'n}\cdot\hat{\mathbf{e}}_{\mathrm{m}})
-\hat{\mathbf{e}}_{\mathrm{m}}}{R_{n'n}} - \frac{e^{-jk_u R_{n'n}}}{4\pi}\,k_u^2\,
\frac{\widehat{\mathbf{R}}_{n'n}\times(\hat{\mathbf{e}}_{\mathrm{m}}\times\widehat{\mathbf{R}}_{n'n})}{R_{n'n}}.
\end{aligned}
\end{equation}
Combining this with $-\hat{\mathbf{e}}_{\mathrm{m}}\nabla_{\mathbf{r}_{n'}}^2\psi
= k_u^2 \hat{\mathbf{e}}_{\mathrm{m}}\psi
= \frac{e^{-jk_u R_{n'n}}}{4\pi}k_u^2\frac{\hat{\mathbf{e}}_{\mathrm{m}}}{R_{n'n}}$
and regrouping the $k_u^2$ terms via
$\hat{\mathbf{e}}_{\mathrm{m}}-\widehat{\mathbf{R}}_{n'n}(\widehat{\mathbf{R}}_{n'n}\cdot\hat{\mathbf{e}}_{\mathrm{m}})
=\widehat{\mathbf{R}}_{n'n}\times(\hat{\mathbf{e}}_{\mathrm{m}}\times\widehat{\mathbf{R}}_{n'n})$,
we arrive at Eq.~\eqref{eq:H}.
The coefficient $3$ emerges from $\nabla_{\mathbf{r}_{n'}}(\widehat{\mathbf{R}}_{n'n}\cdot\hat{\mathbf{e}}_{\mathrm{m}})
=\frac{1}{R_{n'n}}\big(\hat{\mathbf{e}}_{\mathrm{m}}
-(\widehat{\mathbf{R}}_{n'n}\cdot\hat{\mathbf{e}}_{\mathrm{m}})\widehat{\mathbf{R}}_{n'n}\big)$
together with the radial derivative of $\psi$, producing
$3\widehat{\mathbf{R}}_{n'n}(\widehat{\mathbf{R}}_{n'n}\cdot\hat{\mathbf{e}}_{\mathrm{m}})-\hat{\mathbf{e}}_{\mathrm{m}}$.

\section{Derivations for the Jacobian-Aided Coupling-Consistent Hologram Update}
\label{app:jacobian_proof}

\subsection{Exact Jacobian of the coupled RHS operator}

We recall the coupling-aware operator
$
\mathbf{C}_u(\mathbf{m})=(\mathbf{I}-\mathbf{D}(\mathbf{m})\boldsymbol{\Xi}_u)^{-1}
$
and
$
\mathbf{M}_u(\mathbf{m})=\mathbf{C}_u(\mathbf{m})\mathbf{D}(\mathbf{m})\mathbf{F}_u
$.
$\mathbf{C}_u(\mathbf{m}) \triangleq \big(\mathbf{I}-\mathbf{D}(\mathbf{m})\boldsymbol{\Xi}_u\big)^{-1}$,
$\mathbf{M}_u(\mathbf{m}) \triangleq \mathbf{C}_u(\mathbf{m})\mathbf{D}(\mathbf{m})\mathbf{F}_u$,
$\mathbf{A}_u(\mathbf{m})\triangleq \mathbf{I}-\mathbf{D}(\mathbf{m})\boldsymbol{\Xi}_u$,
$\mathbf{C}_u(\mathbf{m})=\mathbf{A}_u(\mathbf{m})^{-1}$,
$
\mathrm{d}\mathbf{C}_u
=\mathrm{d}\big(\mathbf{A}_u^{-1}\big)
=-\mathbf{A}_u^{-1}(\mathrm{d}\mathbf{A}_u)\mathbf{A}_u^{-1}
=\mathbf{C}_u\,\mathrm{d}\!\big(\mathbf{D}(\mathbf{m})\boldsymbol{\Xi}_u\big)\,\mathbf{C}_u 
=\mathbf{C}_u\,(\mathrm{d}\mathbf{D})\boldsymbol{\Xi}_u\,\mathbf{C}_u$,
$\mathrm{d}\mathbf{M}_u
=(\mathrm{d}\mathbf{C}_u)\mathbf{D}\mathbf{F}_u+\mathbf{C}_u(\mathrm{d}\mathbf{D})\mathbf{F}_u
=\mathbf{C}_u(\mathrm{d}\mathbf{D})\big(\boldsymbol{\Xi}_u\mathbf{M}_u+\mathbf{F}_u\big)$,
$\mathbf{D}\triangleq \mathbf{D}(\mathbf{m})$.

Let $\mathbf{e}_n$ be the $n$-th standard basis and define $\mathbf{E}_n\triangleq\mathrm{diag}(\mathbf{e}_n)$ so that
$\frac{\partial \mathbf{D}(\mathbf{m})}{\partial m_n}=\mathbf{E}_n$.
$\frac{\partial \mathbf{C}_u}{\partial m_n}=
\mathbf{C}_u\,\mathbf{E}_n\boldsymbol{\Xi}_u\,\mathbf{C}_u$,
$\frac{\partial \mathbf{M}_u}{\partial m_n}=
\mathbf{C}_u\,\mathbf{E}_n\big(\boldsymbol{\Xi}_u\mathbf{M}_u+\mathbf{F}_u\big)$,
 $n=1,\ldots,N$.

At iteration $t$, denote $\mathbf{m}^{(t)}$, $\mathbf{C}_u^{(t)}=\mathbf{C}_u(\mathbf{m}^{(t)})$,
$\mathbf{M}_u^{(t)}=\mathbf{M}_u(\mathbf{m}^{(t)})$, and $\Delta\mathbf{m}\triangleq \mathbf{m}-\mathbf{m}^{(t)}$.
The first-order Taylor expansion yields the surrogate used in Sec.~\ref{subsec:jacobian_m_update}:
\begin{equation}
\begin{aligned}
\mathbf{M}_u(\mathbf{m})
&=
\mathbf{M}_u^{(t)}+\sum_{n=1}^{N}\Delta m_n\,
\frac{\partial \mathbf{M}_u}{\partial m_n}\bigg|_{\mathbf{m}^{(t)}}
+\mathcal{O}(\|\Delta\mathbf{m}\|^2)\\
&=
\mathbf{M}_u^{(t)}+\sum_{n=1}^{N}\Delta m_n\,
\mathbf{C}_u^{(t)}\mathbf{E}_n\Big(\boldsymbol{\Xi}_u\mathbf{M}_u^{(t)}+\mathbf{F}_u\Big)\\
&\qquad+\mathcal{O}(\|\Delta\mathbf{m}\|^2)\\
&=
\mathbf{M}_u^{(t)}+\mathbf{C}_u^{(t)}\mathbf{D}(\Delta\mathbf{m})
\underbrace{\Big(\boldsymbol{\Xi}_u\mathbf{M}_u^{(t)}+\mathbf{F}_u\Big)}_{\triangleq~\mathbf{T}_u^{(t)}}\\
&\qquad+\mathcal{O}(\|\Delta\mathbf{m}\|^2).
\end{aligned}
\label{eq:app_M_taylor}
\end{equation}

\subsection{Quadratic WMSE in $\Delta\mathbf{m}$ and PSD of $\mathbf{Q}^{(t)}$}

Define the affine surrogate operator
$\widetilde{\mathbf{M}}_u(\mathbf{m}\,|\,\mathbf{m}^{(t)})
\triangleq \mathbf{M}_u^{(t)}+\mathbf{C}_u^{(t)}\mathbf{D}(\Delta\mathbf{m})\mathbf{T}_u^{(t)}$.
For each user--stream--subband triple $(k,i,u)$, define
$z_{k,i,u}(\mathbf{m})\triangleq \mathbf{h}_{k,u}\widetilde{\mathbf{M}}_u(\mathbf{m}|\mathbf{m}^{(t)})\mathbf{v}_{i,u}$.
\begin{equation}
\begin{aligned}
z_{k,i,u}(\mathbf{m})
&=
\mathbf{h}_{k,u}\mathbf{M}_u^{(t)}\mathbf{v}_{i,u}
+\mathbf{h}_{k,u}\mathbf{C}_u^{(t)}\mathbf{D}(\Delta\mathbf{m})\mathbf{T}_u^{(t)}\mathbf{v}_{i,u}\\
&\triangleq
z_{k,i,u}^{(t)}
+\mathbf{r}_{k,u}^{(t)}\mathbf{D}(\Delta\mathbf{m})\mathbf{f}_{i,u}^{(t)}\\
&=
z_{k,i,u}^{(t)}+\big(\mathbf{a}_{k,i,u}^{(t)}\big)^H\Delta\mathbf{m},\\
\end{aligned}
\label{eq:app_z_affine_long}
\end{equation}
$z_{k,i,u}^{(t)}\triangleq \mathbf{h}_{k,u}\mathbf{M}_u^{(t)}\mathbf{v}_{i,u}$,
$\mathbf{r}_{k,u}^{(t)}\triangleq \mathbf{h}_{k,u}\mathbf{C}_u^{(t)}$,
$\mathbf{f}_{i,u}^{(t)}\triangleq \mathbf{T}_u^{(t)}\mathbf{v}_{i,u}$,
$\mathbf{a}_{k,i,u}^{(t)}
\triangleq
\Big((\mathbf{r}_{k,u}^{(t)})^T\odot \mathbf{f}_{i,u}^{(t)}\Big)^*$,
$\Delta\mathbf{m}\in\mathbb{R}^N$.

Then $|z_{k,i,u}(\mathbf{m})|^2$ expands as
\begin{equation}
\begin{aligned}
|z_{k,i,u}(\mathbf{m})|^2
&=
\Big|z_{k,i,u}^{(t)}+\big(\mathbf{a}_{k,i,u}^{(t)}\big)^H\Delta\mathbf{m}\Big|^2\\
&=
|z_{k,i,u}^{(t)}|^2
+2\Re\!\left\{(z_{k,i,u}^{(t)})^*
\big(\mathbf{a}_{k,i,u}^{(t)}\big)^H\Delta\mathbf{m}\right\}\\
&\quad
+\Delta\mathbf{m}^T
\Re\!\left(\mathbf{a}_{k,i,u}^{(t)}\big(\mathbf{a}_{k,i,u}^{(t)}\big)^H\right)
\Delta\mathbf{m}.
\end{aligned}
\label{eq:app_absz_expand_long}
\end{equation}

Substituting \eqref{eq:app_absz_expand_long} into the MSE expression
$
e_{k,u}
=
|g_{k,u}|^2\left(\sum_{i=1}^K|z_{k,i,u}|^2+\sigma_{k,u}^2\right)
-2\Re\{g_{k,u}z_{k,k,u}\}+1
$
and collecting terms in $\Delta\mathbf{m}$ yields the quadratic form:
\begin{equation}
\begin{aligned}
e_{k,u}(\mathbf{m})
&=\Delta\mathbf{m}^T\mathbf{Q}_{k,u}^{(t)}\Delta\mathbf{m}
-2\Re\!\left\{\big(\mathbf{q}_{k,u}^{(t)}\big)^H\Delta\mathbf{m}\right\}+c_{k,u}^{(t)},\\
\mathbf{Q}_{k,u}^{(t)}
&\triangleq
|g_{k,u}|^2\sum_{i=1}^{K}
\Re\!\left(\mathbf{a}_{k,i,u}^{(t)}\big(\mathbf{a}_{k,i,u}^{(t)}\big)^H\right),\\
\mathbf{q}_{k,u}^{(t)}
&\triangleq
g_{k,u}^*\mathbf{a}_{k,k,u}^{(t)}
-|g_{k,u}|^2\sum_{i=1}^{K} z_{k,i,u}^{(t)}\,\mathbf{a}_{k,i,u}^{(t)},\\
c_{k,u}^{(t)}
&\triangleq
|g_{k,u}|^2\left(\sum_{i=1}^{K}|z_{k,i,u}^{(t)}|^2+\sigma_{k,u}^2\right)
-2\Re\!\left\{g_{k,u}z_{k,k,u}^{(t)}\right\}+1.
\end{aligned}
\label{eq:app_mse_quad_long}
\end{equation}

Weighting by $w_{k,u}$ and summing over $(k,u)$ gives
$
\sum_{u=1}^{U}\sum_{k=1}^{K} w_{k,u}e_{k,u}(\mathbf{m})=
\Delta\mathbf{m}^T\mathbf{Q}^{(t)}\Delta\mathbf{m}-2\Re\left\{\big(\mathbf{q}^{(t)}\big)^H\Delta\mathbf{m}\right\}+\mathrm{const},
$
where 
$
\mathbf{Q}^{(t)}\triangleq\sum_{u=1}^{U}\sum_{k=1}^{K} w_{k,u}\mathbf{Q}_{k,u}^{(t)},
$
and 
$
\mathbf{q}^{(t)}\triangleq\sum_{u=1}^{U}\sum_{k=1}^{K} w_{k,u}\mathbf{q}_{k,u}^{(t)}.
$

Finally, $\mathbf{Q}^{(t)}\succeq \mathbf{0}$ follows from (for any $\mathbf{x}\in\mathbb{R}^N$)
\begin{equation}
\begin{aligned}
\mathbf{x}^T\Re\!\left(\mathbf{a}\mathbf{a}^H\right)\mathbf{x}
&=
\Re\!\left(\mathbf{x}^H\mathbf{a}\mathbf{a}^H\mathbf{x}\right)
=
|\mathbf{a}^H\mathbf{x}|^2
\ge 0 \\
&\Rightarrow\ 
\Re\!\left(\mathbf{a}\mathbf{a}^H\right)\succeq \mathbf{0}
\ \Rightarrow\ 
\mathbf{Q}^{(t)}\succeq \mathbf{0}.
\end{aligned}
\label{eq:app_Q_psd_long}
\end{equation}

\subsection{Convex quadratic RHS power constraint under the surrogate}

Let $\mathbf{V}_u=[\mathbf{v}_{1,u},\ldots,\mathbf{v}_{K,u}]$ and define
$\mathbf{X}_u^{(t)}\triangleq \mathbf{M}_u^{(t)}\mathbf{V}_u$,
$\mathbf{Y}_u^{(t)}\triangleq \mathbf{T}_u^{(t)}\mathbf{V}_u$.
Then
$\widetilde{\mathbf{M}}_u(\mathbf{m}|\mathbf{m}^{(t)})\mathbf{V}_u
=\mathbf{X}_u^{(t)}+\mathbf{C}_u^{(t)}\mathbf{D}(\Delta\mathbf{m})\mathbf{Y}_u^{(t)}$.
Using $\|\mathbf{Z}\|_F^2=\mathrm{Tr}(\mathbf{Z}\mathbf{Z}^H)$ yields
\begin{equation}
\begin{aligned}
\mathrm{Tr}\!\Big(\widetilde{\mathbf{M}}_u\mathbf{V}_u\mathbf{V}_u^H\widetilde{\mathbf{M}}_u^H\Big)
&=
\Big\|\mathbf{X}_u^{(t)}+\mathbf{C}_u^{(t)}\mathbf{D}(\Delta\mathbf{m})\mathbf{Y}_u^{(t)}\Big\|_{F}^{2} \\
&=
\|\mathbf{X}_u^{(t)}\|_{F}^{2}
+2\Re\!\left\{\mathrm{Tr}\!\Big((\mathbf{X}_u^{(t)})^H\mathbf{C}_u^{(t)}\mathbf{D}(\Delta\mathbf{m})\mathbf{Y}_u^{(t)}\Big)\right\} \\
&\quad+
\mathrm{Tr}\!\Big(\mathbf{Y}_u^{(t)H}\mathbf{D}(\Delta\mathbf{m})\mathbf{G}_u^{(t)}\mathbf{D}(\Delta\mathbf{m})\mathbf{Y}_u^{(t)}\Big),
\end{aligned}
\label{eq:app_power_expand_1}
\end{equation}
where
$\mathbf{G}_u^{(t)}\triangleq \mathbf{C}_u^{(t)H}\mathbf{C}_u^{(t)}\succeq\mathbf{0}$
and
$\mathbf{S}_u^{(t)}\triangleq \mathbf{Y}_u^{(t)}\mathbf{Y}_u^{(t)H}\succeq\mathbf{0}$.
\begin{equation}
\begin{aligned}
\mathrm{Tr}\!\Big(\mathbf{Y}_u^{(t)H}\mathbf{D}(\Delta\mathbf{m})\mathbf{G}_u^{(t)}
\mathbf{D}(\Delta\mathbf{m})\mathbf{Y}_u^{(t)}\Big)
&=
\mathrm{Tr}\!\Big(\mathbf{D}(\Delta\mathbf{m})\mathbf{G}_u^{(t)}
\mathbf{D}(\Delta\mathbf{m})\mathbf{S}_u^{(t)}\Big) \\
&=
\Delta\mathbf{m}^{T}
\Re\!\Big(\mathbf{G}_u^{(t)}\odot \mathbf{S}_u^{(t)T}\Big)
\Delta\mathbf{m} \\
&\triangleq
\Delta\mathbf{m}^{T}\mathbf{R}_u^{(t)}\Delta\mathbf{m}.
\end{aligned}
\label{eq:app_power_quad}
\end{equation}
\begin{equation}
\begin{aligned}
\mathrm{Tr}\!\Big((\mathbf{X}_u^{(t)})^H\mathbf{C}_u^{(t)}\mathbf{D}(\Delta\mathbf{m})\mathbf{Y}_u^{(t)}\Big)
&=
\mathrm{Tr}\!\Big(\mathbf{D}(\Delta\mathbf{m})
\mathbf{Y}_u^{(t)}(\mathbf{X}_u^{(t)})^H\mathbf{C}_u^{(t)}\Big) \\
&=
\Delta\mathbf{m}^T
\,\mathrm{diag}\!\Big(
\mathbf{Y}_u^{(t)}(\mathbf{X}_u^{(t)})^H\mathbf{C}_u^{(t)}
\Big).
\end{aligned}
\label{eq:app_power_linear}
\end{equation}

Define
$\mathbf{b}_u^{(t)}\triangleq \Re\!\Big(\mathrm{diag}(\mathbf{Y}_u^{(t)}(\mathbf{X}_u^{(t)})^H\mathbf{C}_u^{(t)})\Big)\in\mathbb{R}^N$
and $c_u^{(t)}\triangleq \|\mathbf{X}_u^{(t)}\|_F^2$.
Then the RHS power becomes a convex quadratic function of $\Delta\mathbf{m}$:
\begin{equation}
\begin{aligned}
\mathrm{Tr}\!\Big(\widetilde{\mathbf{M}}_u\mathbf{V}_u\mathbf{V}_u^H\widetilde{\mathbf{M}}_u^H\Big)
&=
\Delta\mathbf{m}^T\mathbf{R}_u^{(t)}\Delta\mathbf{m}
+2\,\mathbf{b}_u^{(t)T}\Delta\mathbf{m}
+c_u^{(t)},\\
\mathbf{R}_u^{(t)}
&\triangleq \Re\!\Big(\mathbf{G}_u^{(t)}\odot \mathbf{S}_u^{(t)T}\Big)\succeq \mathbf{0},\\
\mathbf{x}^T\mathbf{R}_u^{(t)}\mathbf{x}
&=
\mathrm{Tr}\!\Big(\mathbf{D}(\mathbf{x})\mathbf{G}_u^{(t)}\mathbf{D}(\mathbf{x})\mathbf{S}_u^{(t)}\Big)\ge 0, \forall \mathbf{x}\in\mathbb{R}^N,
\end{aligned}
\label{eq:app_power_quad_final_long}
\end{equation}

Therefore, under the surrogate $\widetilde{\mathbf{M}}_u$, the RHS power constraint is the convex QCQP constraint
\begin{equation}
\begin{aligned}
\eta\sum_{u=1}^{U}\Big(
\Delta\mathbf{m}^T\mathbf{R}_u^{(t)}\Delta\mathbf{m}
+2\,\mathbf{b}_u^{(t)T}\Delta\mathbf{m}
+c_u^{(t)}
\Big)
\le P_{\mathrm{RHS}}.
\end{aligned}
\label{eq:app_rhs_power_qc}
\end{equation}

%
%\section*{Biography Section}
%\vspace{-30pt}
%
%\begin{IEEEbiography}[{\includegraphics[width=1in,height=1.3in,clip,keepaspectratio]{./figures/wuliangshun.jpg}}]{Liangshun Wu}  (Member, IEEE) received his B.S. degree in Central South University, Changsha, China, in 2014,  M.S. and Ph.D. from Wuhan University, Wuhan, China, in 2017 and 2021. He was a Visiting Scholar at the University of Electro-Communications, Tokyo, Japan, in 2024. He is currently a postdoctoral researcher at Shanghai Jiao Tong University, Shanghai, China.  
%\end{IEEEbiography}
%%\vfill

\end{document}